\documentclass[11pt,fleqn]{article}
\usepackage{amssymb,latexsym,amsmath,amsfonts,graphicx, ebezier}
\usepackage{pictex,color}

    \advance\textheight by \topskip

    \numberwithin{equation}{section}

\def\beq{\begin{equation}}
\def\eeq{\end{equation}}
\def\nn{\nonumber}
\def\bit{\begin{itemize}}
\def\eit{\end{itemize}}

\makeatletter
\def\eqalign#1{\null\vcenter{\def\\{\cr}\openup\jot\m@th
  \ialign{\strut$\displaystyle{##}$\hfil&$\displaystyle{{}##}$\hfil
      \crcr#1\crcr}}\,}
\makeatother

\newcommand{\be}{\begin{equation}}
\newcommand{\ee}{\end{equation}}

    \def\Re{{\rm Re \,}}
    \def\Im{{\rm Im \,}}
    \def\Ai{{\rm Ai \,}}

    \def\bigO{{\cal O}}

    \def\P2n{{\rm P}_{{\rm II}}^{(n)}}

    \newtheorem{theorem}{Theorem}[section]
    \newtheorem{lemma}[theorem]{Lemma}
    
    \newtheorem{proposition}[theorem]{Proposition}
    
    \newtheorem{Definition}[theorem]{Definition}
    
    \newtheorem{Remark}[theorem]{Remark}
    \newenvironment{remark}{\begin{Remark}\rm}{\end{Remark}}
    \newtheorem{Example}[theorem]{Example}
    
    \newtheorem{Assumptions}[theorem]{Assumptions}

    \newenvironment{proof}%
    {\rm \trivlist \item[\hskip \labelsep{\bf Proof. }]}%
    {\hspace*{\fill}$\Box$\endtrivlist}
    {\rm \trivlist \item[\hskip \labelsep{\bf Proof}]}%
    {\hspace*{\fill}$\Box$\endtrivlist}

    \newcommand{\supp}{{\operatorname{supp}}}

    \DeclareMathOperator*{\Tr}{Tr}

\begin{document}
\title{Biorthogonal ensembles with two-particle interactions}
\author{Tom Claeys\footnote{Institut de Recherche en Math\'ematique et Physique,  Universit\'e
catholique de Louvain, Chemin du Cyclotron 2, B-1348
Louvain-La-Neuve, BELGIUM}
 \ and Stefano Romano\footnotemark[\value{footnote}] }

\maketitle

\begin{abstract}
We investigate determinantal point processes on $[0,+\infty)$ of the form
\begin{equation*}\label{probability distribution}
\frac{1}{Z_n}\prod_{1\leq i<j\leq n}(\lambda_j-\lambda_i)\prod_{1\leq i<j\leq n}(\lambda_j^\theta-\lambda_i^\theta) \prod_{j=1}^n w(\lambda_j)d\lambda_j,\qquad \theta\geq 1.
\end{equation*}
We prove that the biorthogonal polynomials associated to such models satisfy a recurrence relation and a Christoffel-Darboux formula if $\theta\in\mathbb Q$, and that they can be characterized in terms of $1\times 2$ vector-valued Riemann-Hilbert problems which exhibit some non-standard properties. In addition, we obtain expressions for the equilibrium measure associated to our model if $w(\lambda)=e^{-nV(\lambda)}$ in the one-cut case with and without hard edge.

\end{abstract}

\section{Introduction}

We study determinantal point processes which consist of $n$ particles $\lambda_1,\ldots, \lambda_n \in [0,+\infty)$ following a probability distribution of the form
\begin{equation}\label{probability distribution}
\frac{1}{Z_n}\Delta(\lambda_1,\ldots, \lambda_n)\cdot \Delta(\lambda_1^\theta,\ldots, \lambda_n^\theta)\cdot \prod_{j=1}^n w(\lambda_j)d\lambda_j,
\end{equation}
where $\theta\in [1,+\infty)$ and $\Delta$ is the Vandermonde determinant defined as
\begin{equation}\label{VdM}
\Delta(z_1,\ldots, z_n)=\prod_{1\leq i<j\leq n}(z_j-z_i).
\end{equation}
The weight function $w$ is positive and such that $x^kw(x)$ is integrable for any $k\in\mathbb N$.
The partition function $Z_n=Z_n(\theta,w)$ is given by
\begin{equation}
Z_n=\int_{[0,+\infty)^n}\Delta(\lambda_1,\ldots, \lambda_n)\cdot \Delta(\lambda_1^\theta,\ldots, \lambda_n^\theta)\cdot \prod_{j=1}^n w(\lambda_j)d\lambda_j.
\end{equation}
The probability distribution (\ref{probability distribution}) shows a  repulsion between the different particles $\lambda_i$ and a repulsion between the $\lambda_i^\theta$'s.
Such point processes (\ref{probability distribution}) with {\em two-particle interactions} appeared in the study of disordered conductors in the metallic regime \cite{Muttalib}, where it was shown that the correlation kernel can be expressed in terms of biorthogonal polynomials: it is given by
\begin{equation}\label{Kn}
K_n(x,y)=\sum_{j=0}^{n-1}p_j(x)q_j(y^\theta)\sqrt{w(x)w(y)},
\end{equation}
where $p_j(x)=\kappa_jx^j+\ldots$ and $q_j(x)=\kappa_jx^j+\ldots $ are polynomials of degree $j$ characterized by the orthogonality conditions
\begin{equation}\label{orth}
\int_0^{+\infty}p_j(x)q_k(x^\theta)w(x)dx=\delta_{jk},\end{equation}
where we normalize $p_j$ and $q_j$ in such a way that their leading coefficients $\kappa_j$ are equal and positive.
Existence and uniqueness of the families $\{p_j,j=0,1,\ldots\}$ and $\{q_j,j=0,1,\ldots\}$ is guaranteed if the bimoment matrix
\begin{equation}\label{Gramm}
\left(\int_0^{+\infty}x^{k+j\theta}w(x)dx\right)_{j,k=0,\ldots, n-1}
\end{equation}
is nonsingular for every $n$ \cite{IserlesNorsett2, AvHvM, Bertola}. It follows from Proposition \ref{proposition: biop} (i) below that this is indeed the case.
The polynomials $p_j$ and $q_j$ are biorthogonal polynomials of a special type, and
their history goes back to the 1950's when
they appeared in the penetration of gamma rays through matter \cite{SpencerFano} in the Laguerre case
$w(x)=x^\alpha e^{-nx}$ for $\theta=2$. The Laguerre biorthogonal polynomials were subsequently studied for $\theta\in\mathbb N$ by several authors in the 1960's, and recurrence relations, generating functions, differential equations, and a Rodrigues' formula were derived \cite{Preiser, Konhauser2, Carlitz, GeninCalvez2, GeninCalvez1, Prabhakar, Srivastava}.
A slightly more general notion of biorthogonal polynomials was introduced in \cite{Konhauser1}.
Biorthogonal polynomials with respect to a constant weight on $[0,1]$ were studied in \cite{LubinskySidiStahl, LubinskyStahl}, motivated by problems in numerical analysis.
The polynomials $p_j$ and $q_j$ can also be interpreted as multiple orthogonal polynomials \cite{Kuijlaars} of a degenerate type. In particular, $p_j$ is the multiple orthogonal polynomial of type II corresponding to $j$ different orthogonality weights $x^{k\theta}w(x)$, $k=0,\ldots, j-1$.

\medskip

In the simplest case $\theta=1$, if the weight $w$ takes the form $w(x)=x^\alpha e^{-nV(x)}$, the joint probability distribution (\ref{probability distribution}) is realized by the eigenvalues of $n\times n$ positive-definite Hermitian random matrices invariant under unitary conjugation with distribution
\begin{equation}
\frac{1}{\widehat Z_n}(\det M)^\alpha \exp(-n\Tr V(M))dM,\qquad dM=\prod_{j=1}^ndM_{jj}\prod_{1\leq i<j\leq n}d\Re M_{ij} d\Im M_{ij},
\end{equation}
and the correlation kernel is built out of usual orthogonal polynomials. In this case, large $n$ asymptotic results have been obtained for general classes of external fields $V$.
For $\theta=2$, the process (\ref{probability distribution}) appears in the so-called O(n) model \cite{Kostov, EynardZinnJustin} with ${\rm n}=-2$ and the limiting macroscopic behavior for a large number of particles has been studied, see also \cite{LSZ} for the Laguerre case.

\medskip

Except for $\theta=1$, asymptotic results about the microscopic behavior of particles are available only for special choices of weight functions: Borodin \cite{Borodin} obtained asymptotics for the correlation kernels for Laguerre, Jacobi, and Hermite weights for general $\theta\geq 1$ and computed large $n$ scaling limits of the correlation kernel (\ref{Kn}) for $x$ and $y$ close to $0$. He showed that there is a class of limiting kernels, depending on $\alpha$ and $\theta$, expressed in terms of certain generalizations of Bessel functions. For general $V$ and $\theta\notin\{1,2\}$, even the macroscopic behavior of the particles as $n\to\infty$ has not been described before.

\medskip

In the first part of the paper, we recall some well-known properties of the biorthogonal polynomials which hold for general weight functions $w$: they can be expressed as a determinant and they have a multiple integral formula.
In addition, we prove that the biorthogonal polynomials satisfy recurrence relations and a Christoffel-Darboux type formula if $\theta$ is rational:
\begin{theorem}\label{theorem: biOP}
Let $\theta=\frac{a}{b}$ with $a,b\in\mathbb N$, $a>b$.
The biorthogonal polynomials $p_j, q_j$ defined by (\ref{bi-orth}) satisfy recurrence relations of the form
\begin{align}\label{rec0}
x^ap_k(x) &=u_0(k)p_{k+a}(x) + u_1(k)p_{k+a-1}(x)+\dots +u_{a+b}(k)p_{k-b}(x),\\
x^bq_k(x) &=v_0(k)q_{k+b}(x) + v_1(k)q_{k+b-1}(x)+\dots +v_{a+b}(k)q_{k-a}(x),\label{rec2}
\end{align}
where we use the convention that $p_j(x)=q_j(x)=0$ for $j<0$,
and the coefficients $u_j(k), v_j(k)$ are related by
\beq\label{u_uhat0}
u_j(k) =v_{a+b-j}(k+a-j).
\eeq
Moreover, we have the Christoffel-Darboux formula
\begin{multline}\label{CD}
\sum_{k=0}^{n-1}p_k(x)q_k(y^\theta)=\frac{1}{x^a-y^a}\left(
\sum_{\ell=1}^{a}\sum_{k=n-\ell}^{n-1}u_{a-\ell}(k)p_{k+\ell}(x)q_k(y^\theta)\right. \\
\left. -\sum_{\ell=1}^b\sum_{k=n-\ell}^{n-1}u_{a+\ell}(k+\ell)p_k(x)q_{k+\ell}(y^\theta)
\right).\end{multline}
\end{theorem}
\begin{remark}
If $\theta$ is integer, a similar Christoffel-Darboux formula had already been obtained in \cite{Ilyasov}. In \cite{IserlesNorsett2}, Christoffel-Darboux formulas of a different nature have been derived for biorthogonal polynomials.
It is also interesting to compare our Christoffel-Darboux formula with the ones for multiple orthogonal polynomials \cite{DaemsKuijlaars} and for multiple orthogonal polynomials of mixed type \cite{DaemsKuijlaars2}.
\end{remark}

As a next result,
we characterize the polynomials $p_j$ and $q_j$ in terms of Riemann-Hilbert (RH) problems, which are similar to RH problems in \cite{ClaeysWang} for polynomials related to a random matrix model with equi-spaced external source. We believe that asymptotic methods similar to the one developed in \cite{ClaeysWang} can be applied to the RH problems in order to obtain asymptotics for the biorthogonal polynomials, but such an asymptotic RH analysis is not the aim of the present paper.

For the formulation of the RH problems, we restrict ourselves to weights of the form $w(x)=x^\alpha e^{-nV(x)}$ here, with $\alpha>-1$ and $V$ sufficiently smooth on $[0,+\infty)$. Nevertheless, it is straightforward to generalize the RH problems for general weight functions by adapting condition (d) below.
Define the following modified Cauchy transform of the polynomial $p_j$:
\begin{equation} \label{def Cp0}
  C p_j(z)  \equiv  \frac{1}{2\pi i} \int_{0}^{+\infty} \frac{p_j(x)}{x^\theta - z^\theta} w(x) dx.
\end{equation}
This function is well-defined as long as $z^\theta\in\mathbb C\setminus [0,+\infty)$, but we consider it as a function
defined in $\mathbb H_\theta\setminus[0,+\infty)$, with
\begin{equation}\label{Htheta}
\mathbb{H}_{\theta} = \left\{ z\in \mathbb{C}:-\frac{\pi}{\theta}<\text{arg}(z)<\frac{\pi}{\theta}\right\}.
\end{equation}
For $z\in \mathbb H_\theta$, we define $z^\theta$ corresponding to arguments between $-\pi/\theta$ and $\pi/\theta$.
In the following theorem, we characterize $p_j$ and $Cp_j$ in terms of a RH problem.

\begin{theorem}\label{theorem: RHP1}
Let $Y$ be defined by
\begin{equation}\label{def Y}
  Y(z) \equiv  \left(\frac{1}{\kappa_j}p_j(z), \frac{1}{\kappa_j}C p_j(z)\right).
\end{equation} Then $Y$ is the unique function which satisfies the following conditions:
\subsubsection*{RH problem for $Y$}
\begin{itemize}
\item[(a)] \label{enu:RH_of_Y:1}
  $Y = (Y_1, Y_2)$ is analytic in $\left(\mathbb C, \mathbb H_\theta \setminus \mathbb [0,+\infty)\right)$,
\item[(b)] \label{enu:RH_of_Y:2}
  $Y$ has continuous boundary values $Y_\pm$ when
  approaching $(0,+\infty)$ from above ($+$) and below ($-$), and we have
  \begin{equation}\label{RHP_Y:jump1}
    Y_+(x) = Y_-(x)
    \begin{pmatrix}
      1 & \frac{1}{\theta x^{\theta-1}}w(x) \\
      0 & 1
    \end{pmatrix},
    \quad \textnormal{for $x>0$,}
  \end{equation}
  \item[(c1)]
    as $z\to\infty$, $Y_1$ behaves as  $Y_1(z)=z^j+\bigO(z^{j-1})$,
  \item[(c2)]
    as $z\to\infty$ in $\mathbb H_\theta$, $Y_2$ behaves as
    $Y_2(z)=\bigO(z^{-(j+1)\theta})$,
   \item[(d)] as $z\to 0$, $Y_2(z)=\bigO(1)+\bigO(z^{\alpha+1-\theta})$ if $\alpha+1-\theta\neq 0$; $Y_2(z)=\bigO(\log z)$ if $\alpha+1-\theta=0$,
   \item[(e)] for $x>0$, we have the boundary condition
   \begin{equation}
   Y_2(e^{\pi i/\theta}x)=   Y_2(e^{-\pi i/\theta}x).
   \end{equation}
\end{itemize}
\end{theorem}
\begin{remark}
There are several differences between the above RH problem and the classical RH problem for orthogonal polynomials introduced by Fokas, Its, and Kitaev \cite{FokasItsKitaev}. A first difference is that our RH problem is vector-valued and not matrix-valued. One could add a second row to $Y$ containing the orthogonal polynomial of degree $j-1$ and its modified Cauchy transform, thus obtaining a $2\times 2$ matrix RH problem, but contrary to the determinant of the solution to the Fokas-Its-Kitaev problem, the determinant of $Y$ would not be equal to $1$ in our case, and without the guarantee that the determinant of $Y$ is different from $0$, it is not clear that a matrix RH problem is more convenient than a vector RH problem, for example for asymptotic analysis. Therefore we prefer to write our RH problem as a vector-valued one.
Another, more crucial, difference is that the entries $Y_1$ and $Y_2$ live in different domains $\mathbb C$ and $\mathbb H_\theta\setminus[0,+\infty)$. The jump contour $[0,+\infty)$ lies in the intersection of both domains, and this leads to a standard multiplicative jump condition with matrix-valued jump matrix.
The asymptotic conditions for $Y_1$ and $Y_2$ hold as $z\to\infty$ and $z\to 0$ in their respective domains. The boundary condition (e) glues together the edges of the sector $\mathbb H_\theta$, so that $Y_2$ actually lives on a cone.
\end{remark}

A similar RH characterization exists for the polynomials $q_j$.
Define
\begin{equation} \label{def C}
  \widetilde C q_j(z)  \equiv  \frac{1}{2\pi i} \int_{0}^{+\infty} \frac{q_j(x^\theta)}{x- z} w(x) dx,\qquad z\in \mathbb C\setminus[0,+\infty).
\end{equation}

\begin{theorem}\label{theorem: RHP2}
Let $\widetilde Y$ be defined by
\begin{equation}\label{def tilde Y}
  \widetilde Y(z) \equiv  \left(\frac{1}{\kappa_j}q_j(z^\theta), \frac{1}{\kappa_j} \widetilde C q_j(z)\right).
\end{equation} Then $\widetilde Y$ is the unique function which satisfies the following conditions:
\subsubsection*{RH problem for $\widetilde Y$}
\begin{itemize}
\item[(a)] \label{enu:RH_of_Y:12}
  $\widetilde Y = (\widetilde Y_1, \widetilde Y_2)$ is analytic in $(\mathbb H_\theta, \mathbb C \setminus \mathbb [0,+\infty))$,
\item[(b)] \label{enu:RH_of_Y:22}
  $\widetilde Y$ has continuous boundary values $\widetilde Y_\pm$ when
  approaching $(0,+\infty)$ from above and below, and we have
  \begin{equation}\label{RHP_Ytilde:jump1}
    \widetilde Y_+(x) = \widetilde Y_-(x)
    \begin{pmatrix}
      1 & w(x) \\
      0 & 1
    \end{pmatrix},
    \quad \textnormal{for $x >0$,}
  \end{equation}
  \item[(c1)]
    as $z\to\infty$ in $\mathbb H_\theta$, $\widetilde Y_1$ behaves as  $\widetilde Y_1(z)=z^{j\theta}+\bigO(z^{(j-1)\theta})$,
  \item[(c2)]
    as $z\to\infty$, $\widetilde Y_2$ behaves as
    $\widetilde Y_2(z)=\bigO(z^{-(j+1)})$,
   \item[(d1)] as $z\to 0$ in $\mathbb H_\theta$, $\widetilde Y_1(z)=\bigO(1)$,
   \item[(d2)] as $z\to 0$, $\widetilde Y_2(z)=\bigO(1)+\bigO(z^{\alpha})$ if $\alpha\neq 0$; $\widetilde Y_2(z)=\bigO(\log z)$ if $\alpha=0$,
   \item[(e)] for $x>0$, we have the boundary condition
   \begin{equation}
   \widetilde Y_1(e^{\pi i/\theta}x)=  \widetilde Y_1(e^{-\pi i/\theta}x).
   \end{equation}
\end{itemize}
\end{theorem}

\begin{remark}
In the RH problem for $q_j$, the first entry $\widetilde Y_1$ lives on a cone, and the second entry $\widetilde Y_2$ lives in the complex plane. Consequently, we have a boundary condition for $\widetilde Y_1$, and the asymptotic conditions for $\widetilde Y_1$ are valid in $\mathbb H_\theta$.
\end{remark}

In the final part of the paper, we will study in detail the equilibrium problem associated to the point processes (\ref{probability distribution}) if $w$ has the form $w(x)= e^{-nV(x)}$, with $V$ sufficiently smooth on $[0,+\infty)$ and satisfying the growth condition
\begin{equation}\label{growth}
\lim_{x\to +\infty}\frac{V(x)}{\log x}=+\infty.
\end{equation}
The relevant equilibrium problem for our model is to find the unique probability measure $\mu=\mu_{V,\theta}$ minimizing the functional
\begin{equation}\label{energy}I_{V,\theta}(\mu)=\frac{1}{2}\iint\frac{1}{|x-y|}d\mu(x)d\mu(y)+\frac{1}{2}\iint\log\frac{1}{|x^\theta-y^\theta|}d\mu(x)d\mu(y)+\int V(x)d\mu(x),
\end{equation}
among all probability measures $\mu$ on $[0,+\infty)$.
The solution to this problem describes the limiting mean distribution of the particles in the process as $n\to\infty$ \cite{ESS}.
Existence and uniqueness of the measure $\mu$ can be proved using similar methods as in \cite{Deift}.
If a measure $\mu$ satisfies the Euler-Lagrange conditions
\begin{align}
&\label{EL1} \int\log|x-y|d\mu(y)+\int\log|x^\theta-y^\theta|d\mu(y)-V(x)=\ell,&&\mbox{ for $x\in\supp\,\mu$,}\\
& \label{EL2}\int\log|x-y|d\mu(y)+\int\log|x^\theta-y^\theta|d\mu(y)-V(x)\leq \ell,&&\mbox{ for $x\in [0,+\infty)$,}
\end{align}
for some $\ell\in\mathbb R$, then it is well understood that $\mu$ is the unique minimizer of (\ref{energy}).
We will study in detail, for $\theta\geq 1$, a class of external fields $V$ for which the equilibrium measure is supported on a single interval of the form $[0,b]$ (the one-cut case with a hard edge) and we will also study the case where the support consists of one interval of the form $[a,b]$ with $a>0$ (the one-cut case without hard edge).
This equilibrium problem has been studied extensively for $\theta=1$, see e.g.\ \cite{Deift, DKM, SaffTotik}, and many results are available about the equilibrium measure in this case. Once the support of the equilibrium measure is known, one can compute the density of $\mu_{V,1}$ from the Euler-Lagrange condition (\ref{EL1}) using the Sokhotski-–Plemelj formula. For $\theta>1$, it is not obvious how one can generalize this method, but we will show that one can compute the equilibrium measure $\mu$ by solving (\ref{EL1}) if the measure $\mu$ is supported on one interval.

\medskip

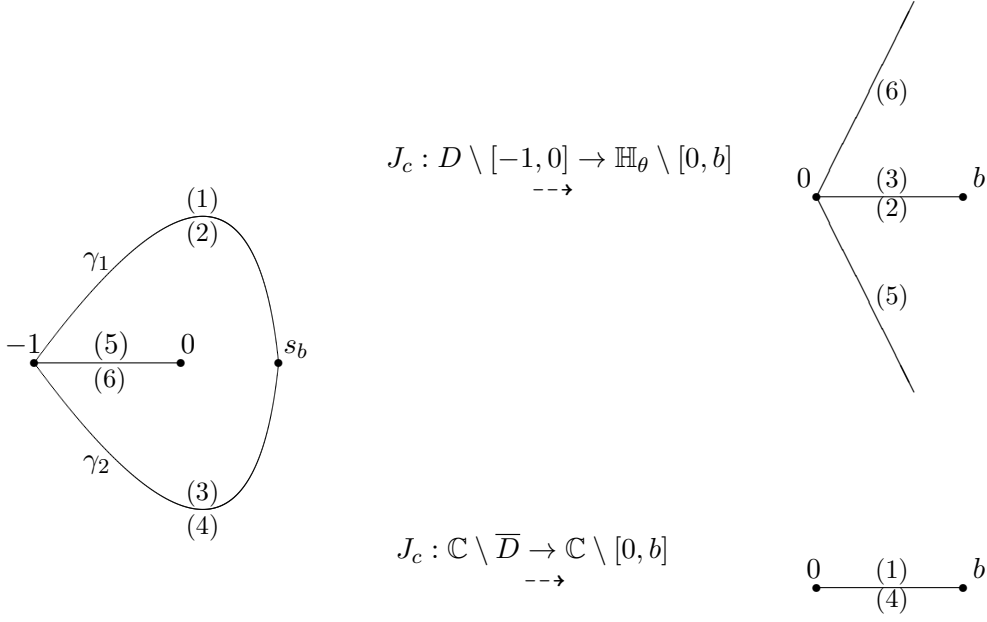
\begin{figure}[t]
\begin{center}
    \setlength{\unitlength}{1.3truemm}
    \begin{picture}(100,73)(0,25)
        \put(0,63){\thicklines\circle*{.8}}
        \put(15,63){\thicklines\circle*{.8}}
        \put(25,63){\thicklines\circle*{.8}}
        \put(0,63){\line(1,0){15}}
        \qbezier(0,63)(22,93)(25,63)
        \qbezier(0,63)(22,33)(25,63)
\put(-3,64){$-1$}
\put(15,64){$0$}
\put(25.5,64){$s_b$}
\put(5,73){$\gamma_1$}
\put(5,52.2){$\gamma_2$}
\put(6,64){$(5)$}
\put(6,60.5){{\small $(6)$}}
        \put(80,80){\line(1,2){10}}
        \put(80,80){\line(1,-2){10}}
        \put(80,80){\line(1,0){15}}
        \put(80,80){\thicklines\circle*{.8}}
        \put(95,80){\thicklines\circle*{.8}}
    \put(15.5,79){{\small $(1)$}}
    \put(15.5,45.5){{\small $(4)$}}
\put(15.5,75.5){{\small $(2)$}}
    \put(15.5,49){{\small $(3)$}}

\put(78,81){$0$}
\put(96,81){$b$}
\put(50,80){${\bf \dashrightarrow}$}
\put(36,83){$J_c:D\setminus[-1,0]\to \mathbb H_\theta\setminus[0,b]$}

        \put(80,40){\line(1,0){15}}
        \put(80,40){\thicklines\circle*{.8}}
        \put(95,40){\thicklines\circle*{.8}}

\put(79,41){$0$}
\put(96,41){$b$}
\put(50,40){${\bf \dashrightarrow}$}

\put(37,43){$J_c:\mathbb C\setminus \overline{D}\to \mathbb C\setminus[0,b]$}

 \put(86,81){{\small $(3)$}}
    \put(86,78){{\small $(2)$}}
    \put(86,90){{\small $(6)$}}
    \put(86,69){{\small $(5)$}}
\put(86,38){{\small $(4)$}}
    \put(86,41){{\small $(1)$}}

    \end{picture}
    \caption{The transformation $J_c$ which maps $D\setminus[-1,0]$ to $\mathbb H_\theta\setminus[0,b]$ and $\mathbb C\setminus \overline D$ to $\mathbb C\setminus[0,b]$. The upper and lower edges (1)-(6) of the boundary of $D\setminus[-1,0]$ at the left are mapped to (1)-(6) at the right.}
    \label{figure: J}
\end{center}
\end{figure}

A key role in the construction of the equilibrium measure in the one-cut case with hard edge is played by the transformation
\begin{equation}\label{def J}
J(s)=J_c(s)=c\left(s+1\right)\left(\frac{s+1}{s}\right)^{\frac{1}{\theta}},
\end{equation}
where the branch is chosen which is analytic in $\mathbb C\setminus [-1,0]$, and such that $J(s)\sim cs$ as $s\to\infty$.
It has two critical points on the real line, $-1$ and $s_b = \frac{1}{\theta}$, which are mapped to $0$ and
\begin{equation}\label{b}
b = c\frac{(1+\theta)^{1+1/\theta}}{\theta}.
\end{equation}
$J_c(s)$ is real for $s \in(-\infty, -1]\,\cup\, [s_b, +\infty)$ and along two complex conjugate arcs $\gamma_1, \gamma_2$ joining $-1$ and $s_b$ in the upper and lower half planes respectively. Both $\gamma_1$ and $\gamma_2$ are mapped bijectively to the interval $[0, b]$. We denote by $I_+(x)$ and $I_-(x)$ the inverse images of $x\in[0, b]$ belonging to $\gamma_1, \gamma_2$ respectively. We write $\gamma$ for the union of $\gamma_1$ and $\gamma_2$, oriented in the counterclockwise direction, and we write $D$ for the region enclosed by $\gamma$. As we will show in Section \ref{section: J}, the region of the complex plane outside $\gamma$ is mapped bijectively to the complex plane with the interval $[0, b]$ removed; the inner region $D\setminus[-1,0]$ is mapped bijectively to $\mathbb H_\theta\setminus [0,b]$.
The mapping $J_c$ is illustrated in Figure \ref{figure: J}.
\begin{remark}
If $\theta=1$, it is easy to verify that the curve $\gamma$ is the unit circle, and that $J$ maps both the interior and the exterior of the circle to $\mathbb C\setminus [0,b]$.
\end{remark}

\begin{theorem}\label{theorem: one-cut reg}
Let $\theta>1$, let $V$ be twice continuously differentiable on $[0,+\infty)$ and such that (\ref{growth}) holds.
If
\bit
\item[(i)] $V''(x)x+V'(x)> 0$ for all $x>0$,
\item[(ii)] $V''(x)\geq 0$ for all $x>0$,
\eit
the equilibrium measure $\mu_{V,\theta}$ is supported on a single interval $[0, b]$, where $b$ is given by
$$
b = c\frac{(1+\theta)^{1+1/\theta}}{\theta},
$$
and $c$ is the unique solution of the equation
\beq\label{getc0}
\frac{1}{2\pi i} \oint_{\gamma}\frac{V'(J_c(s))J_c(s)}{s}ds = 1+\theta.
\eeq
The density $\psi=\psi_{V,\theta}$ of $\mu_{V,\theta}$ is given by
\beq\label{psi}
\psi(x) = \frac{1}{2\pi^2 x}\int_0^b\Big(V''(y)y + V'(y)\Big)\log\left|\frac{I_+(y) - I_-(x)}{I_+(y)- I_+(x)}\right|dy.
\eeq
\end{theorem}
\begin{remark}\label{remark: local}
From the formula (\ref{psi}) for the equilibrium density, we can derive its local behavior near the endpoints.
We have
\begin{align}
&J_c(s)=ce^{\mp \pi i /\theta}(s+1)^{1+\frac{1}{\theta}}\left(1+\bigO(s+1)\right),&& s\to -1,\ \pm\Im s>0, \\
&J_c(s)= b+\frac{1}{2}J_c''(s_b)(s-s_b)^2+\bigO\left((s-s_b)^3\right),&& s\to s_b,
\end{align}
where the branch cut for $(s+1)^{1+\frac{1}{\theta}}$ lies on $(-\infty, -1]$.
For $x\in (0,b)$, the inverses $I_\pm$ behave consequently as
\begin{align}
&I_\pm(x)=-1+c^{-\frac{\theta}{\theta+1}}e^{\pm \frac{\pi i}{\theta+1}}x^{\frac{\theta}{\theta+1}}(1+\bigO(x)),&& x\to 0_+,\\
&I_\pm(x)=s_b\pm i\sqrt\frac{2}{J_c''(s_b)}(b-x)^{1/2}(1+\bigO(x-b)),&& x\to b_-.
\end{align}
Substituting this into (\ref{psi}) as $x\to b_-$,
we obtain after a straightforward calculation
\begin{align}
&\label{exponent}\psi(x)=d_1x^{-\frac{1}{\theta+1}}+\bigO(x^{\frac{\theta-1}{\theta+1}}),&& x\to 0,\\
&\label{square root}\psi(x)=d_2(b-x)^{1/2} \left(1+o(1)\right),&& x\to b_-,
\end{align}
with
\begin{align}
&d_1=-\frac{1}{\pi^2}c^{-\frac{\theta}{\theta+1}}\sin\frac{\pi}{\theta+1}\int_0^b\Big(V''(y)y + V'(y)\Big)\Im\frac{1}{I_+(y)+1}dy>0,\label{c1}\\
&d_2=-\frac{1}{\pi^2b}\sqrt\frac{2}{J_{c}''(s_b)}\int_0^b \Big(V''(y)y + V'(y)\Big)\Im\frac{1}{I_+(y)-s_b}dy>0.
\end{align}
\end{remark}
\begin{remark}
The square root behavior (\ref{square root}) near $b$ is similar to what one has in the random matrix case $\theta=1$ near soft edges of the support.
On the other hand, the exponent $-\frac{1}{\theta+1}$ in (\ref{exponent}) near $0$ is different from the behavior of the equilibrium measure in random matrix ensembles, where one typically has an exponent $-1/2$ near a hard edge.
In more general random matrix ensembles, one can have equilibrium measures corresponding other rational exponents $\pm p/q$ \cite{BE}, but the exponents $-1/(\theta+1)$ we obtain for our model, range over the entire interval $[-1/2,0)$.
\end{remark}

The following result states that the formula (\ref{psi}) for the equilibrium density holds under weaker assumptions than the ones in Theorem \ref{theorem: one-cut reg}.

\begin{theorem}\label{theorem: eq measure formula}
Let $\theta>1$, let $V$ be twice continuously differentiable on $[0,+\infty)$ and such that (\ref{growth}) holds.
If a probability measure $\mu$ is supported on one interval $[0,b]$, has a continuous density $\psi$ on $(0,b)$, and satisfies the Euler-Lagrange conditions (\ref{EL1})-(\ref{EL2}), then $\psi=\psi_{V,\theta}$ is given by (\ref{psi}). The endpoint $b$ is given by (\ref{b}), where $c$ solves equation (\ref{getc0}).
\end{theorem}

In order to study the one-cut case without hard edge, define the transformation
\begin{equation}\label{def tilde J}
\widetilde J(s)=\widetilde J_{c_0,c_1}(s)=(c_1s+c_0)\left(\frac{s+1}{s}\right)^{\frac{1}{\theta}},\qquad c_0>c_1>0,
\end{equation}
with branch cut on $[-1,0]$, and such that $\widetilde J(s)\sim c_1s$ as $s\to\infty$.
$\widetilde J$ has two critical points $s_a< s_b$ on the real line,
\begin{align}\label{sab}
s_a &=-\frac{\theta-1}{2\theta}-\frac{1}{2\theta c_1}\sqrt{4c_0c_1\theta+c_1^2(\theta-1)^2},\\
s_b &=-\frac{\theta-1}{2\theta}+\frac{1}{2\theta c_1}\sqrt{4c_0c_1\theta+c_1^2(\theta-1)^2}.\label{sab2}
\end{align}
We have $s_a<-1$ and $s_b>0$, and $a\equiv \widetilde J(s_a)>0$, $b\equiv \widetilde J(s_b)>a$.
There are two complex conjugate curves $\widetilde\gamma_1, \widetilde\gamma_2$ starting at $s_a$ and ending at $s_b$ in the upper and lower half plane respectively which are mapped to the interval $[a, b]$ through $\widetilde J_{c_0, c_1}$. Let $\widetilde\gamma$ be the counterclockwise oriented closed curve consisting of the union of $\widetilde\gamma_1$ and $\widetilde\gamma_2$, enclosing a region $\widetilde D$.
Define $\widetilde I_+, \widetilde I_-$ as the two inverses of $\widetilde J_{c_0, c_1}$ from $[a, b]$ to $\widetilde \gamma_1, \widetilde \gamma_2$ respectively.
$\widetilde J_{c_0,c_1}$ maps the region of the complex place outside $\widetilde\gamma$ to $\mathbb C\setminus [a, b]$. The inner region with $[-1,0]$ removed is mapped to $\mathbb H_\theta\setminus [a,b]$.
The mapping $\widetilde J_{c_0,c_1}$ is illustrated in Figure \ref{figure: tilde J}.

\begin{figure}[t]
\begin{center}
    \setlength{\unitlength}{1.3truemm}
    \begin{picture}(100,73)(0,25)
        \put(0,63){\thicklines\circle*{.8}}
        \put(15,63){\thicklines\circle*{.8}}
        \put(25,63){\thicklines\circle*{.8}}
        \put(-6,63){\thicklines\circle*{.8}}
        \put(0,63){\line(1,0){15}}
        \qbezier(-6,63)(-4,83)(10,83)
        \qbezier(-6,63)(-4,43)(10,43)
        \qbezier(10,83)(23,83)(25,63)
        \qbezier(10,43)(23,43)(25,63)
\put(-2,64){$-1$}
\put(15,64){$0$}
\put(26,63.5){$s_b$}
\put(-9,64){$s_a$}
\put(1,77){$\widetilde\gamma_1$}
\put(1,46){$\widetilde\gamma_2$}
\put(6,64){$(5)$}
\put(6,60.5){{\small $(6)$}}
        \put(80,80){\line(1,2){10}}
        \put(80,80){\line(1,-2){10}}
        \put(90,80){\line(1,0){15}}
        \put(90,80){\thicklines\circle*{.8}}
        \put(105,80){\thicklines\circle*{.8}}
        \put(80,80){\thicklines\circle*{.8}}
    \put(14,83){{\small $(1)$}}
    \put(14,41.5){{\small $(4)$}}
\put(14,79.2){{\small $(2)$}}
    \put(14,45){{\small $(3)$}}

\put(78,81){$0$}
\put(105,81){$b$}
\put(89,81){$a$}
\put(50,80){${\Large{\bf \dashrightarrow}}$}
\put(37,83){$\widetilde J:D\setminus[-1,0]\to \mathbb H_\theta\setminus[a,b]$}

        \put(90,40){\line(1,0){15}}
        \put(90,40){\thicklines\circle*{.8}}
        \put(105,40){\thicklines\circle*{.8}}

\put(89,41){$a$}
\put(105,41){$b$}
\put(50,40){${\Large{\bf \dashrightarrow}}$}
\put(37,43){$\widetilde J:\mathbb C\setminus \overline{D}\to \mathbb C\setminus[a,b]$}

 \put(96,81){{\small $(3)$}}
    \put(96,77.5){{\small $(2)$}}
    \put(86,90){{\small $(6)$}}
    \put(86,69){{\small $(5)$}}
\put(96,37.5){{\small $(4)$}}
    \put(96,41){{\small $(1)$}}

    \end{picture}
    \caption{The transformation $\widetilde J=\widetilde{J}_{c_0,c_1}$ which maps $D\setminus[-1,0]$ to $\mathbb H_\theta\setminus[a,b]$ and $\mathbb C\setminus \overline D$ to $\mathbb C\setminus[a,b]$. The upper and lower edges (1)-(6) of the boundary of $D\setminus[-1,0]$ at the left are mapped to (1)-(6) at the right.}
    \label{figure: tilde J}
\end{center}
\end{figure}

\begin{theorem}\label{theorem: eq measure formula soft edge}
Let $\theta>1$, let $V$ be twice continuously differentiable on $[0,+\infty)$ and such that (\ref{growth}) holds.
If a probability measure $\mu$ is supported on one interval $[a,b]$ with $a>0$, has a continuous density $\psi$ on $(a,b)$, and satisfies the Euler-Lagrange conditions (\ref{EL1})-(\ref{EL2}), then $\psi=\psi_{V,\theta}$ is given by
\begin{equation}\label{psi2}
\psi(x) = \frac{1}{2\pi^2 x}
\int_a^b\Big(V''(y)y + V'(y)\Big)
\log\left|\frac{\widetilde I_+(y) - \widetilde I_-(x)}{\widetilde I_+(y)- \widetilde I_+(x)}\right| dy.
\end{equation}
The endpoints $a$ and $b$ are given by $a=\widetilde J(s_a)$, $b=\widetilde J(s_b)$, with $s_a$ and $s_b$ given by (\ref{sab})-(\ref{sab2}), and where $c_0$ and $c_1$ solve the system of equations
\begin{align}
&\label{eqc0}\frac{1}{2\pi i}\oint_{\widetilde\gamma}\frac{V'\left(\widetilde J_{c_0, c_1}(t)\right)\widetilde J_{c_0, c_1}(t)}{t}dt=1+\theta,\vspace{3pt}\\
&\label{eqc1}\frac{1}{2\pi i}\oint_{\widetilde\gamma}\frac{V'\left(\widetilde J_{c_0, c_1}(t)\right)\widetilde J_{c_0, c_1}(t)}{t+1}dt=1.
\end{align}
\end{theorem}
\begin{remark}
By a similar computation as in Remark \ref{remark: local}, one observes that the equilibrium density vanishes generically as a square root near the endpoints $a$ and $b$.
\end{remark}

\subsubsection*{Outline}
In Section \ref{section: biOP}, we study the biorthogonal polynomials $p_j$ and $q_j$, and prove the recurrence relation and Christoffel-Darboux formula stated in Theorem \ref{theorem: biOP}. In Section \ref{section: RHP}, we prove the RH characterizations for the biorthogonal polynomials given in Theorem \ref{theorem: RHP1} and Theorem \ref{theorem: RHP2}. In Section \ref{section: eq}, we construct the equilibrium measure $\mu_{V,\theta}$ in the one-cut case with hard edge and without hard edge, which leads to proofs of Theorem \ref{theorem: one-cut reg}, Theorem \ref{theorem: eq measure formula}, and Theorem \ref{theorem: eq measure formula soft edge}. We also study several examples to illustrate the construction and the behavior of the equilibrium measure in particular cases, including examples for the one-cut case with hard edge, the one-cut case without hard edge, and a transition between them. In Section \ref{section: univ} finally, we discuss, without proving any results, various scaling limits of the correlation kernel.

\section{General properties of the bi-orthogonal polynomials $p_j$ and $q_j$}\label{section: biOP}

\subsection{Determinantal and integral expressions for the polynomials}
Let $\theta\geq 1$ and let $w$ be positive on $(0,+\infty)$ and such that $x^kw(x)$ is integrable for every $k\in\mathbb N\cup\{0\}$. Then all the bimoments
\begin{equation}\label{muij}
m_{jk}=\int_0^{+\infty}x^{k+j\theta}w(x)dx, \qquad j,k\in\mathbb N\cup\{0\}
\end{equation}
exist,
and we define
\begin{equation}\label{def Hn}
H_n(w)=\det(m_{jk})_{j,k=0,\ldots, n-1}.
\end{equation}
Write \[\Delta_n(x)=\Delta(x_1,\ldots, x_n),\qquad \Delta_n(x^\theta)=\Delta(x_1^\theta,\ldots, x_n^\theta).\]
In the following proposition, we collect some well-known results about the determinant of $H_n(w)$ and about the biorthogonal polynomials $p_j$ and $q_j$.

\begin{proposition}\label{proposition: biop}
\begin{itemize}
\item[(i)] The determinant $H_n(w)$ can be expressed as
\begin{equation}\label{Hnformula}
H_n(w)=\frac{1}{n!}\int_{[0,+\infty)^n}\Delta_n(x)\cdot \Delta_n(x^\theta) \prod_{j=1}^nw(x_j)dx_j>0.
\end{equation}
\item[(ii)] The polynomials $p_j$ and $q_j$ have the following determinantal expressions:
\begin{equation}\label{det pj}
p_j(x)=\frac{1}{\sqrt{H_j(w)H_{j+1}(w)}}\det\begin{pmatrix}m_{00}&m_{01}&m_{02}&\ldots &m_{0j}\\
m_{10}&m_{11}&m_{12}&\ldots &m_{1j}\\
\vdots & \vdots &\vdots & &\vdots \\
m_{j-1,0}&m_{j-1,1}&m_{j-1,2}&\ldots &m_{j-1,j}\\
1&x&x^2&\ldots &x^j
\end{pmatrix},
\end{equation}
and\begin{equation}\label{det qj}
q_j(x)=\frac{1}{\sqrt{H_j(w)H_{j+1}(w)}}\det\begin{pmatrix}m_{00}&m_{01}&\ldots &m_{0,j-1} &1\\
m_{10}&m_{11}&\ldots &m_{1,j-1}&x\\
\vdots & \vdots & & \vdots&\vdots \\
m_{j,0}&m_{j,1}&\ldots & m_{j,j-1}&x^j
\end{pmatrix}.
\end{equation}
\item[(iii)] The polynomials $p_j$ and $q_j$ have the following multiple integral representations:
\begin{align}
&\label{int pj}p_j(x)=\frac{1}{j!\sqrt{H_j(w)H_{j+1}(w)}}\int_{[0,+\infty)^j}\prod_{k=1}^j(x-x_j)\Delta_j(x) \Delta_j(x^\theta)\prod_{k=1}^jw(x_k)dx_k,\\
&\label{int qj}q_j(x^\theta)=\frac{1}{j!\sqrt{H_j(w)H_{j+1}(w)}}\int_{[0,+\infty)^j}\prod_{k=1}^j(x^\theta-x_j^\theta)\Delta_j(x) \Delta_j(x^\theta) \prod_{k=1}^jw(x_k)dx_k.
\end{align}
\end{itemize}
\end{proposition}

Those properties are straightforward to prove in analogy to similar properties for usual orthogonal polynomials, see for example \cite{Johansson} for (i), \cite{AvHvM, Bertola, DF, IserlesNorsett, Kuijlaars} for (ii) and \cite{DF, Kuijlaars} for (iii). It is also worth noting that it has been proved in \cite{IserlesNorsett} that the polynomials $p_j$ have positive and simple zeros, and that the zeros of $p_j$ and $p_{j-1}$ interlace.

\subsection{Recurrence relation and Christoffel-Darboux formula for $\theta\in\mathbb Q$}
Assume that $\theta=a/b$ with $a> b$ positive integers.
Let $p_j$ and $q_j$ be the biorthogonal polynomials defined by (\ref{orth}). After the substitution $u=x^{1/b}$ in (\ref{orth}), we obtain
\begin{equation}\label{orth2}
\int_0^{+\infty}p_j(u^b)q_k(u^a)\widetilde w(u)du=\delta_{jk},\qquad \widetilde w(u)=bu^{b-1}w(u^b).
\end{equation}
Writing
\begin{equation}
<f,g>\equiv \int_0^{+\infty}f(u^b)g(u^a)\widetilde w(u)du,
\end{equation}
we have the relation
\begin{equation}\label{blockHankel}
< x^a f(x), g(x)> = < f(x), x^bg(x)>,
\end{equation}
and the semi-infinite moment matrix \[\mu_{\infty} \equiv (\mu_{ij})_{i, j=0,1,\ldots}\equiv (<x^i,x^j>)_{i, j=0,1,\ldots}\] is block-Hankel with rectangular $a\times b$ blocks.
In terms of the shift operator $\Lambda = (\delta_{i+1, j})_{i,j=0,1,\ldots}$, we have
\beq\label{blockmu}
\Lambda^a\mu_{\infty} = \mu_{\infty}\left(\Lambda^t\right)^b.
\eeq
Define the semi-infinite column vectors
of biorthogonal polynomials
\begin{align}
&\mathbf{p}(x) = (p_0(x), p_1(x), p_2(x), \dots)^t,\\&\mathbf{q}(x) = (q_0(x), q_1(x), q_2(x), \dots)^t.
\end{align}
Those vectors can be expressed in the basis of monomials as
\begin{equation}\label{vec_bi-orth}
\mathbf{p}(x)=S_1\chi(x)\qquad\mathbf{q}(x)= \left(S_2^{-1}\right)^t\chi(x),
\end{equation}
where $\chi(x) = (1, x, x^2, \dots)^t$, where $S_1$ is a lower triangular semi-infinite matrix, and $S_2$ is an upper triangular semi-infinite matrix.
Since \beq\label{bi-orth}
\langle p_j, q_k\rangle = \delta_{jk},
\eeq
the semi-infinite moment matrix $\mu_\infty$ can be written in the factorized form
\beq\label{Borel}
\mu_\infty = S_1^{-1}S_2.
\eeq
The dressed shift operators
\begin{equation}\label{L}
L_1 = S_1\Lambda S_1^{-1},\qquad
L_2 = S_2\Lambda^tS_2^{-1}
\end{equation}
play the role of recursion operators for the biorthogonal polynomials. Indeed from \eqref{vec_bi-orth} we obtain
\beq\label{matrix_rec}
L_1\mathbf{p}(x) = x\mathbf{p}(x)\qquad  L_2^t \mathbf{q}(x) = x \mathbf{q}(x).
\eeq
The recursion operators \eqref{L} satisfy the identity
\beq\label{bi-graded}
L_1^a = L_2^b.
\eeq
This follows from \eqref{blockmu} and \eqref{Borel}:
$$
L_1^a = S_1\Lambda^aS_1^{-1} = S_1\Lambda^a\mu_{\infty}S_2^{-1} = S_1\mu_{\infty}\left(\Lambda^t\right)^bS_2^{-1} = S_2\left(\Lambda^t\right)^bS_2^{-1} = L_2^b.
$$
In order to prove the recursion relation (\ref{rec0}), note that
$$
\langle x^a p_n(x), x^k\rangle = \langle p_n(x), x^{k+b}\rangle = 0 \qquad \text{for}\;\; k< n-b.
$$
Therefore the expression of $x^ap_n(x)$ as a linear combination of left-orthogonal polynomials $p_j$ can only involve polynomials of degree at least $n-b$. Noting that the left hand side is of degree $n+a$, we obtain \eqref{rec0}. The recursion relation for the right-orthogonal polynomials $q_j$ is obtained in the same way.\\
To prove the relation \eqref{u_uhat0} between the coefficients of the recursion relations, note first by (\ref{rec0})-(\ref{rec2}) that
\begin{equation}\label{uv}
u_j(k)=\langle x^a p_k(x), q_{k+a-j}(x)\rangle,\qquad v_j(k)=\langle p_{k+b-j}(x), x^b q_k(x)\rangle.\end{equation}
On the other hand we have that
\begin{equation}
x^ap_k(x)=\left(L_1^a\mathbf{p}(x)\right)_k,\qquad x^bq_k(x)=\left(L_2^{bt}\mathbf{q}(x)\right)_k,
\end{equation}
and by (\ref{uv}) we obtain
\begin{equation}\label{uv2}
u_j(k)=\left(L_1^a\right)_{k,k+a-j},\qquad v_j(k)=\left(L_2^b\right)_{k+b-j,k}.\end{equation}
By (\ref{bi-graded}), we have $u_j(k)=v_{a+b-j}(k+a-j)$.
This proves the recurrence relation (\ref{rec0}) and the relation between the recurrence coefficients (\ref{u_uhat0}).

\medskip

The Christoffel-Darboux formula is a consequence of the recurrence relation: consider $(x^a-y^b)\sum_{k=0}^{n-1}p_k(x)q_k(y)$ and use the recurrence relations (\ref{rec0}). Cancellations take place by (\ref{u_uhat0}) and one obtains the formula
\begin{multline}
\sum_{k=0}^{n-1}p_k(x)q_k(y)=\frac{1}{x^a-y^b}\left(
\sum_{\ell=1}^{a}\sum_{k=n-\ell}^{n-1}u_{a-\ell}(k)p_{k+\ell}(x)q_k(y)\right. \\
\left. -\sum_{\ell=1}^b\sum_{k=n-\ell}^{n-1}u_{a+\ell}(k+\ell)p_k(x)q_{k+\ell}(y)
\right).\end{multline}
Substituting $y^\theta$ for $y$, we arrive at (\ref{CD}).
This completes the proof of Theorem \ref{theorem: biOP}.

\section{RH problems for $p_j$ and $q_j$}\label{section: RHP}

As stated in Theorem \ref{theorem: RHP1} and Theorem \ref{theorem: RHP2}, the polynomials $p_j$ and $q_j$ can be characterized in terms of two $1\times 2$ vector RH problems. Those RH problems have some non-standard properties compared to those for orthogonal polynomials. It is our hope that the RH problems are suitable for a Deift/Zhou asymptotic analysis following the methods developed in \cite{ClaeysWang}.

\medskip

\subsection{Proof of Theorem \ref{theorem: RHP1}}
Assume that $w(x)=x^\alpha e^{-nV(x)}$ with $\alpha>-1$, and that $\theta>1$.
Define, like in (\ref{def Cp0}),
\begin{equation} \label{def Cp}
  C f(z)  \equiv  \frac{1}{2\pi i} \int_{0}^{+\infty} \frac{f(x)}{x^\theta - z^\theta} w(x) dx,
\end{equation}
defined in $\mathbb H_\theta\setminus[0,+\infty)$ and with $z^\theta$ corresponding to arguments between $-\pi/\theta$ and $\pi/\theta$.

We first show that $Y=\begin{pmatrix}\frac{1}{\kappa_j}p_j& \frac{1}{\kappa_j}Cp_j\end{pmatrix}$ indeed solves the RH problem for $Y$ stated in Theorem \ref{theorem: RHP1}.
Conditions (a) and (e) are clearly satisfied by the definition (\ref{def Cp}).
Condition (b) follows from the relation
\begin{equation}\label{jump C}
 (C p_j)_+(x)-(C p_j)_-(x)= \frac{1}{\theta x^{\theta -1}}p_j(x)w(x),\qquad x\in (0,+\infty),
\end{equation}
which is easily seen by Cauchy's theorem.
Condition (c1) is valid since $Y_1$ is a monic polynomial of degree $j$.
For condition (c2), note that as $z\to\infty$ in $\mathbb H_\theta\setminus[0,+\infty)$,
\begin{equation}\label{Cpas}
Cp_j(x)=-\frac{1}{2\pi i }\sum_{k=1}^{j-1}\frac{1}{z^{(k+1)\theta}}\int_0^{+\infty}p_j(x)x^{k\theta}w(x)dx +\bigO(z^{-(j+1)\theta}),
\end{equation}
and by orthogonality each of the integrals in the sum vanishes. Condition (d) is straightforward to verify using (\ref{def Cp}) and the fact that $w(x)=x^\alpha e^{-nV(x)}$.

\medskip

For the proof of the uniqueness of the RH solution, let us assume that $Y=(Y_1,Y_2)$ is a solution to the RH problem for $Y$, a priori possibly different from  (\ref{def Y}). Then, since $Y_1$ is an entire function, by the asymptotic behavior (c1) and by Liouville's theorem, $Y_1$ is a monic polynomial of degree $j$. Given $Y_1$, the jump condition (b) gives us the relation
\begin{equation}\label{jump Y2}
Y_{2,+}(x)-Y_{2,-}(x)= \frac{1}{\theta x^{\theta -1}}Y_1(x)w(x),\qquad x\in (0,+\infty).
\end{equation}
Consider the function
\[U(s)=Y_2(s^{1/\theta})-(C Y_1)(s^{1/\theta}).\]
Then $U$ is analytic in $\mathbb C\setminus \mathbb R$.
By (\ref{jump C}), which holds for any polynomial $p_j$, and (\ref{jump Y2}), it follows that $U$ is continuous across $(0,+\infty)$, and by condition (e) it is continuous on $(-\infty,0)$. As $s\to 0$, we have that $U(s)=o(s^{-1})$ by condition (d) and (\ref{def Cp}), and it follows that $U$ is entire. Moreover, it tends to $0$ as $s\to \infty$ by condition (c) and (\ref{def Cp}), and it follows that $U=0$, and thus $Y_2=C Y_1$.

It remains to show that $Y_1$ is orthogonal with respect to $w(x),x^\theta w(x),\ldots, x^{(j-1)\theta}w(x)$. Therefore, recall (\ref{Cpas}), which holds for any polynomial $p_j$. Applying this to the polynomial $Y_1$, (\ref{Cpas}) can only be consistent with condition (c2) if
\[\int_0^{+\infty}Y_1(x)x^{k\theta}w(x)dx=0,\qquad k=0,1,\ldots, j-1.\]
By uniqueness of the biorthogonal polynomial $p_j$ we have
$Y_1=\frac{1}{\kappa_j}p_j$ and $Y_2=\frac{1}{\kappa_j}C p_j$.

\subsection{Proof of Theorem \ref{theorem: RHP2}}

The proof of Theorem \ref{theorem: RHP2} is similar to the one of Theorem \ref{theorem: RHP1}.
The fact that $\widetilde Y$ solves the RH problem stated in Theorem \ref{theorem: RHP2}, follows again from a straightforward verification of all of the properties, where the orthogonality of $q_j$ is needed for condition (c2).
The uniqueness follows from similar considerations as before: first one uses conditions (a),  (b), (c1), and (e) to show that $\widetilde Y_1$ is a monic polynomial of degree $j$ in $z^\theta$, then one shows that $\widetilde Y_2=\widetilde C \widetilde Y_1$ using conditions (b), (d), and (e), and finally one shows that $\widetilde Y_1$ is biorthogonal using condition (c2).

\section{Equilibrium problem}\label{section: eq}
\subsection{The mapping $J$}\label{section: J}

Recall the transformation $J=J_c$ defined in (\ref{def J}).
As already explained in the introduction, $J$ has two critical points on the real line, $0$ and $s_b = \frac{1}{\theta}$, which are mapped to $-1$ and
\begin{equation}\label{b2}
b = c\frac{(1+\theta)^{1+1/\theta}}{\theta}.
\end{equation}
Let us now investigate which parts of the complex plane are mapped to $[0,b]$.
Writing $s=re^{i\phi}$ with $0\leq \phi<2\pi$, we have
\[\arg J(s)=(1+\frac{1}{\theta})\arg ( 1 + re^{i\phi})-\frac{1}{\theta}\phi.\] For this to vanish, we need that
\[(1+\frac{1}{\theta})\arg (1+r\cos\phi  + ir\sin\phi)=\frac{1}{\theta}\phi.\]
For $0<\phi<\pi$, the left hand side of this equation is an increasing function of $r$, which is $0$ at $r=0$, and which tends to $(1+\frac{1}{\theta})\phi$ as $r\to\infty$. By continuity, there exists, for any $\phi\in (0,\pi)$, a unique value of $r=r(\phi)$ such that $s=r(\phi)e^{i\phi}$ is mapped to the positive half-line by $J$. The equation $r=r(\phi)$ describes a curve $\gamma_1$ in the upper half plane connecting $-1$ with $s_b$, which is mapped bijectively to $(0,b)$ by $J$, and the same is true for $\gamma_2=\overline{\gamma_1}$. It is also easy to see that $\gamma$, the counterclockwise oriented union of $\gamma_1$ and $\gamma_2$, lies inside the unit circle for $\theta>1$, and that it is independent of $c>0$. The curve $\gamma$ is shown in Figure \ref{figure: J1} for different values of $\theta$. The half-line $(s_b,+\infty)$ is mapped bijectively to $(b,+\infty)$, and
the upper and lower side of the interval $[-1,0]$ are mapped bijectively to the boundary of the region $\mathbb H_\theta$ defined in (\ref{Htheta}).
\begin{figure}[t]
\begin{center}
\includegraphics[scale=0.5,angle=0]{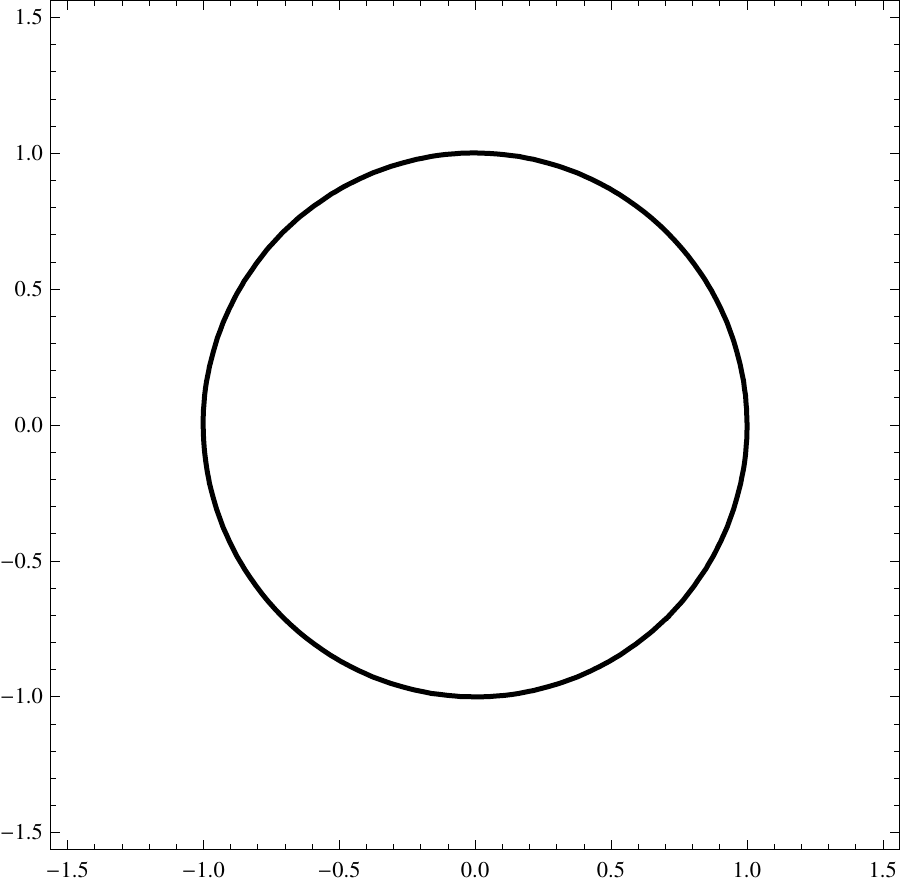}
\includegraphics[scale=0.5,angle=0]{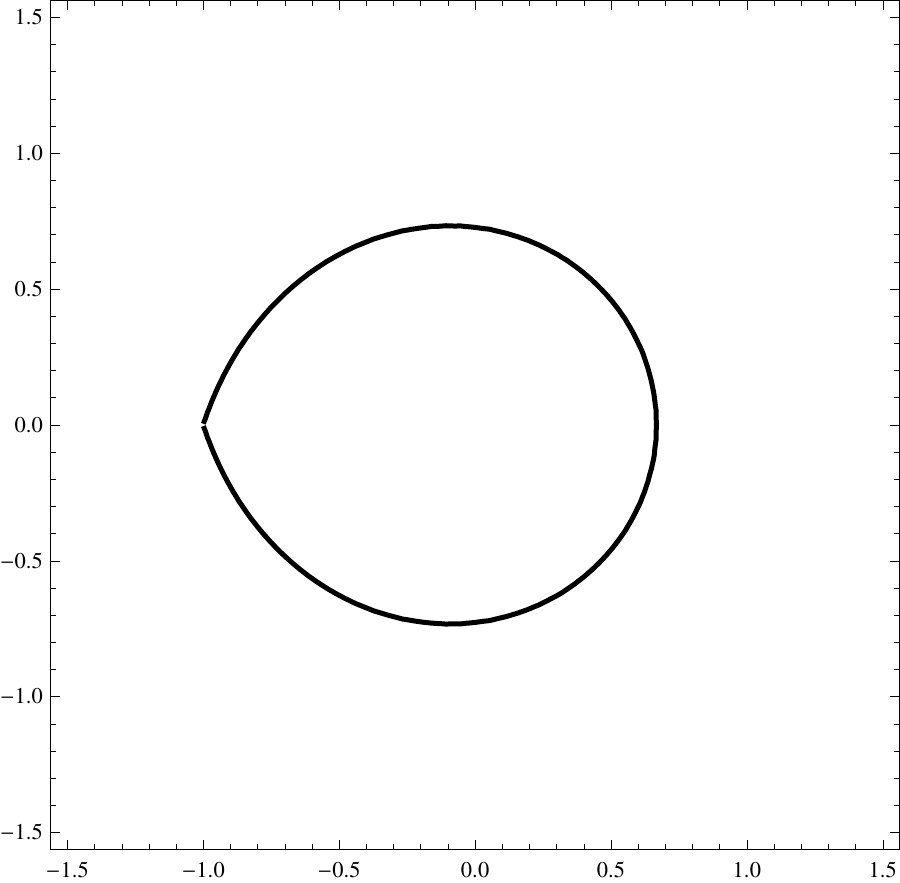}
\includegraphics[scale=0.5,angle=0]{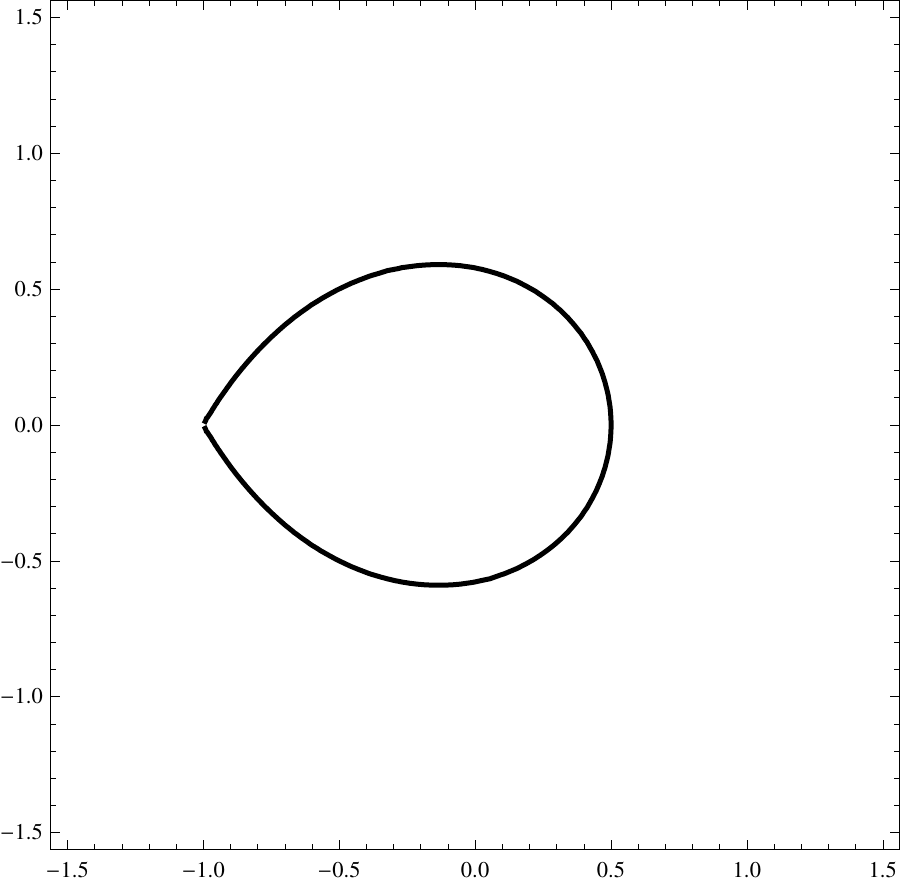}
\includegraphics[scale=0.5,angle=0]{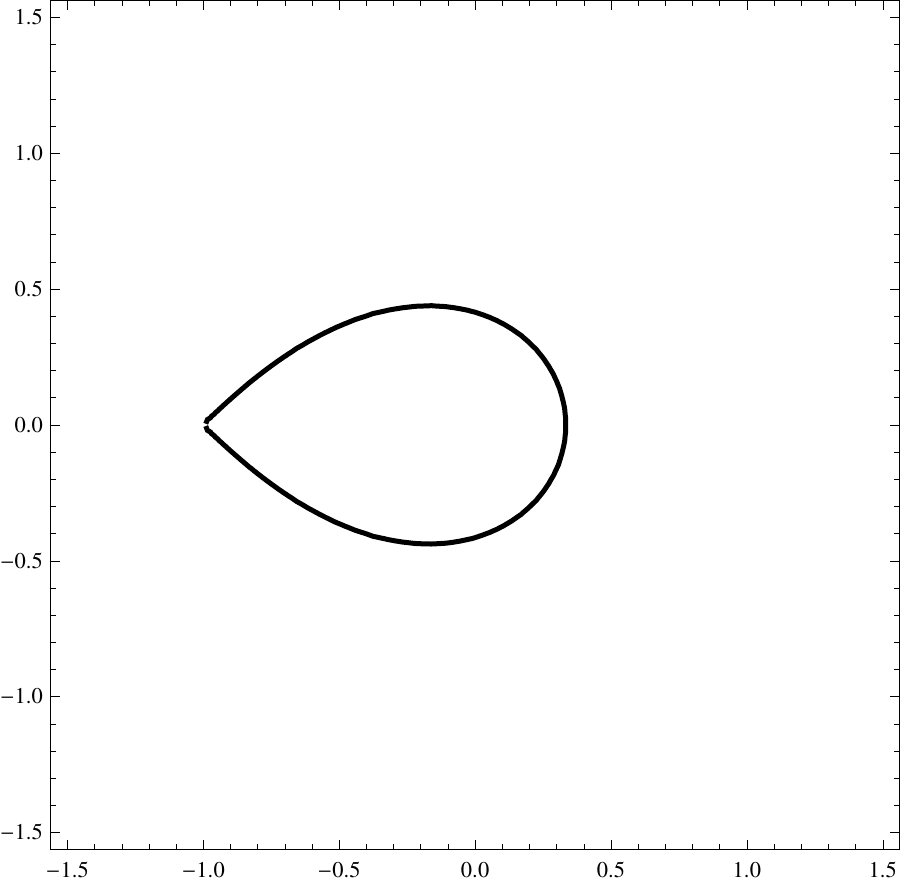}
\includegraphics[scale=0.5,angle=0]{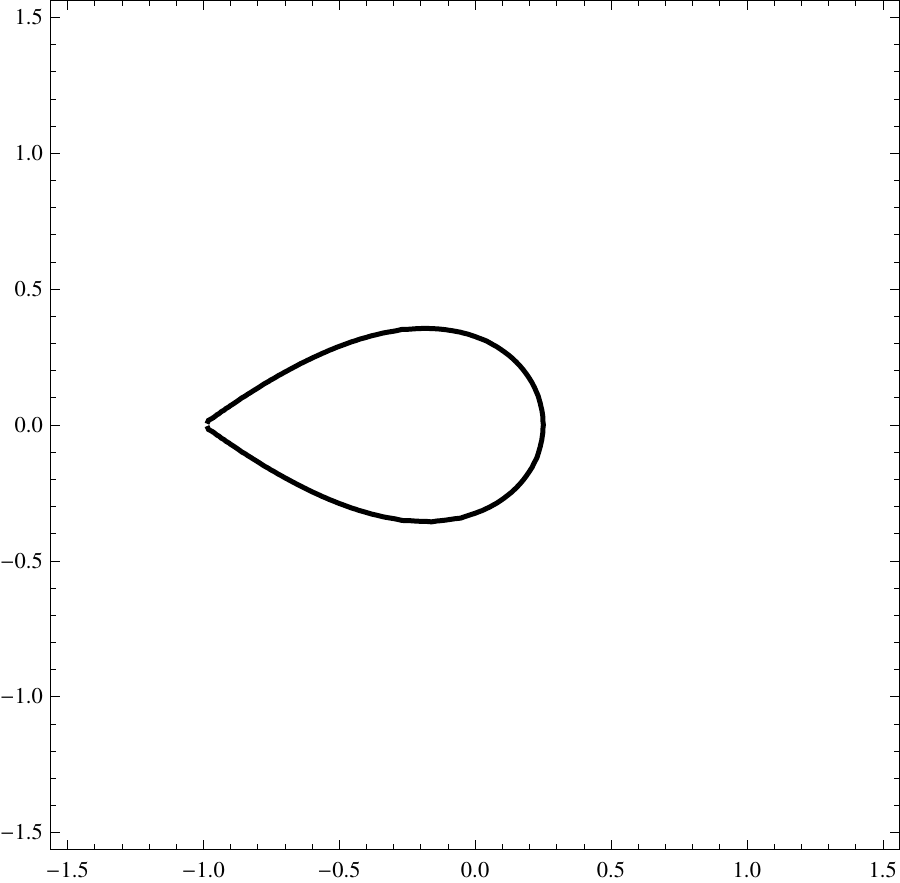}
\includegraphics[scale=0.5,angle=0]{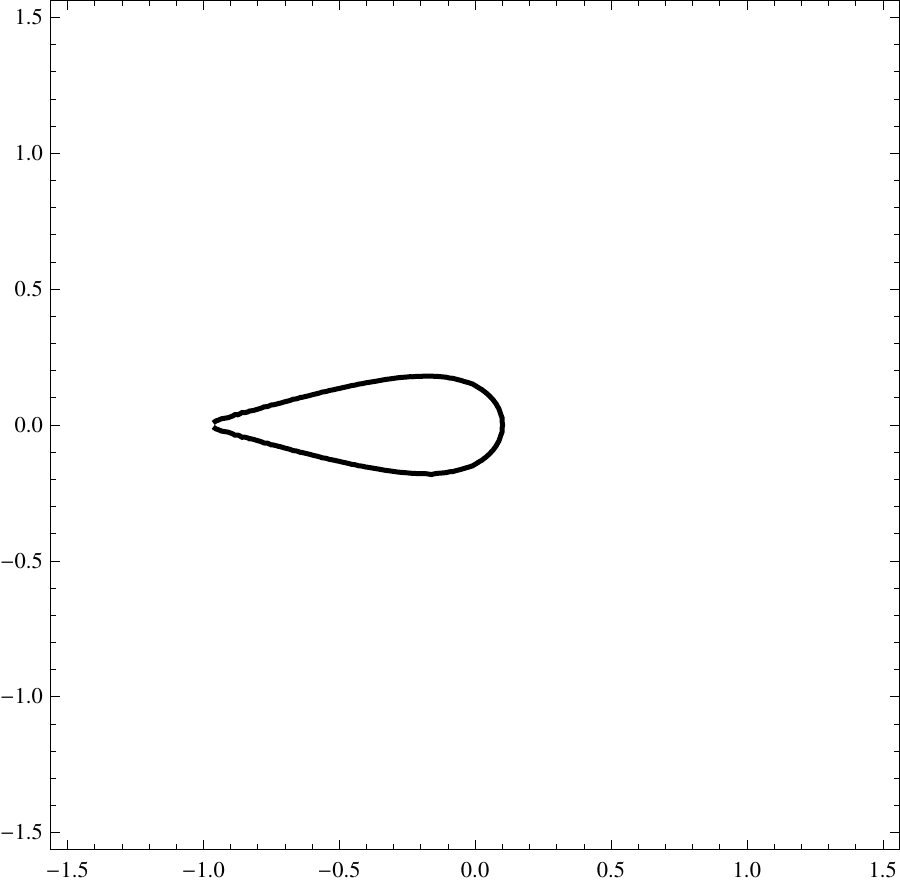}
\end{center}
\caption{The curve $\gamma$ for $\theta=1$ (top left), $\theta=1.5$, $\theta=2$, $\theta=3$, $\theta=4$, and $\theta=10$ (bottom right).}
\label{figure: J1}
\end{figure}

We thus observe that
$J$ maps the region $D\setminus[-1,0]$ to a subset of $\mathbb H_\theta\setminus [0,b]$ and that it is bijective between the boundaries of the regions.
Let $a\in \mathbb H_\theta\setminus[0,b]$ and let $\Sigma$ be a counterclockwise oriented closed curve in $D\setminus[-1,0]$ sufficiently close to its boundary for $J$ to be bijective on $\Sigma$, and such that $J(\Sigma)$ winds around $a$. Then
$\frac{1}{2\pi i}\int_\Sigma \frac{J'(s)}{J(s)-a}ds=\int_{J(\Sigma)}\frac{dz}{z-a}$.
The right hand side of this equation is equal to $2\pi i$, and the left hand side is equal to the number of points $s$ inside $\Sigma$ such that $J(s)=a$ by the residue theorem. This number has to be $1$, and by letting $\Sigma$ approach the boundary of $D\setminus[-1,0]$, one can conclude that there is a unique point in $D\setminus[-1,0]$ which is mapped to $a$. Thus, $J$ is bijective.
In order to prove bijectivity between $\mathbb C\setminus \overline D$ and $\mathbb C\setminus[0,b]$, one can proceed in a similar way. This gives us the following lemma:
\begin{lemma}\label{lemma J}
Let $D$ be the open region enclosed by the curves $\gamma_1$ and $\gamma_2$. $J$ maps $\mathbb C\setminus \overline D$ bijectively to $\mathbb C\setminus [0,b]$, and it maps $D\setminus[-1,0]$ bijectively to $\mathbb H_\theta\setminus[0,b]$.
\end{lemma}

\subsection{Solution in the one-cut case with hard edge}\label{section: eq hard}
In this section, we will prove Theorem \ref{theorem: eq measure formula} and Theorem \ref{theorem: one-cut reg}.
We first need the following lemma.
\begin{lemma} Let $V$ be twice continuously differentiable on $[0,+\infty)$ and let $V$ satisfy (\ref{growth}) and the conditions
\beq\label{nc}
V''(x)x+V'(x)>0,\quad V''(x)\geq 0\qquad\text{for } x>0.
\eeq
Then there exists a unique $c\in (0, +\infty)$ such that
\beq\label{getc}
f(c)\equiv \frac{1}{2\pi i} \oint_{\gamma}\frac{U_{c}(s)}{s}ds = 1+\theta,
\eeq
where $$
U_c(s) = V'(J_c(s))J_c(s).
$$
\end{lemma}
\begin{proof}
We easily verify that $f(c)$ is real for $c\in(0,+\infty)$ and that
$\lim_{c\to 0} f(c)=0$. The latter follows from the fact that $\lim_{c\to 0}U_c(s)=0$ uniformly on $\gamma$.
If we can prove that $f'(c)>0$ for $c>0$ and that $f'(c)>\delta>0$ for $c>0$ large, it follows that $f(c)$ is a bijection from $[0,+\infty)$ to $[0,+\infty)$, which proves the lemma.
We have
\begin{eqnarray*}
f'(c)&=&\frac{d}{dc}\frac{1}{2\pi i}\oint_\gamma\frac{U_c(s)}{s}ds 	\\&=& \frac{1}{2\pi i}\oint_\gamma\Big(V''(J_c(s))J_c(s)+V'(J_c(s))\Big)\frac{J_c(s)}{c}d\log s\\
												&=&-\frac{1}{\pi c}\int_{\gamma_1}\Big(V''(J_c(s))J_c(s)+V'(J_c(s))\Big)J_c(s)\text{ Im }d\log s\\
												&=&-\frac{1}{\pi c}\int_0^{c\frac{(1+\theta)^{1+1/\theta}}{\theta}}\Big(V''(x)x+V'(x)\Big)x\;d\;\text{arg}\,I_+(x).
\end{eqnarray*}
In the second line we used the fact that $\gamma_1$ and $\gamma_2$ are complex conjugates and that $V(J_c(s))$ and $J_c(s)$ are real on $\gamma$. In the last line we made the change of variable $s = I_+(x)$. By inspection of the function $J_c$ one can check that $\text{arg}\,I_+(x)$ strictly decreases from $\pi$ to $0$ along $\gamma_1$. Using \eqref{nc} we conclude that indeed $f'(c)>0$. Moreover, by the convexity of $V$ and (\ref{growth}), we have that $V''(x)$ and $V'(x)$ are positive for large $x$, and hence $V''(x)x+V'(x)$ is larger than a positive constant for $x$ large. This implies that  $f'(c)>\delta>0$ for $c$ sufficiently large,
 which completes the proof.
\end{proof}

\subsubsection{Proof of Theorem \ref{theorem: eq measure formula}}
In this section, we assume that $\mu$ is a probability measure supported on $[0,b]$ with a continuous density $\psi$ and which satisfies the Euler-Lagrange conditions (\ref{EL1})-(\ref{EL2}).
We introduce the logarithmic transforms
\begin{align}
&\label{g}g(z)\equiv\int_0^b\log(z-y)d\mu(y),&z \in{\mathbb{C}\setminus(-\infty, b]},\\
&\label{tilde g}\widetilde{g}(z)\equiv\int_0^b\log(z^\theta - y^\theta)d\mu(y),&z\in\mathbb{H}_\theta\setminus[0, b],
\end{align}
where we choose the logarithms corresponding to arguments between $-\pi$ and $\pi$.
By (\ref{EL1}) and the fact that $\mu$ is a probability measure, it follows that $g$ and $\widetilde g$ satisfy the following conditions.
\subsubsection*{RH problem for $g$ and $\widetilde g$}
\bit
\item[(a)] $(g, \widetilde{g})$ is analytic in $(\mathbb{C}\setminus [-\infty, b], \mathbb{H}_\theta\setminus[0, b])$.
\item[(b)] Writing $g_+, g_-, \widetilde g_+, \widetilde g_-$ for the boundary values of $g$ and $\widetilde g$ when approaching $(-\infty,b)$ and $(0,b)$ from above (for $+$) or below (for $-$), we have the relations
\begin{align}&g_\pm(x)+\widetilde{g}_\mp(x) = V(x)+\ell,&\mbox{ for $x\in (0, b)$},\\
&\widetilde{g}(e^{-i\pi/\theta}x) = \widetilde{g}(e^{i\pi/\theta}x)-2\pi i,&\mbox{ for $x>0$,}\\
&g_+(x)=g_-(x)+2\pi i,&\mbox{ for $x<0$}.
\end{align}
\item[(c1)] As $z\to\infty$, $g(z) = \log z+\bigO(z^{-1})$.
\item[(c2)] As $z\to \infty$ in the sector $\mathbb H_\theta$, we have
	$\widetilde{g}(z) = \theta\log z+\bigO(z^{-\theta})$.
\eit
Let
\begin{equation}\label{G}
G(z) = g'(z)\qquad \widetilde{G}(z) = \widetilde{g}'(z).
\end{equation}
Then the RH conditions for $(g,\widetilde g)$ imply
\subsubsection*{RH problem for $(G,\widetilde G)$}
\bit
\item[(a)] $(G, \widetilde{G})$ is analytic in $(\mathbb{C}\setminus [0, b], \mathbb{H}_\theta\setminus[0, b])$.
\item[(b)] $G_\pm(x)+\widetilde{G}_\mp(x) = V'(x)$ for $x\in (0, b)$.\\
	$\widetilde{G}(e^{-i\pi/\theta}x) = e^{2\pi i/\theta}\widetilde{G}(e^{i\pi/\theta}x)$ for $x>0$.
\item[(c1)] $G(z) = \frac{1}{z} + O(z^{-2})$ for $z\to\infty$.
\item[(c2)]
	$\widetilde{G}(z) = \frac{\theta}{z} + O(z^{-1-\theta})$ for $z\to\infty$ in $\mathbb H_\theta$.
\eit
A crucial observation to solve this RH problem, is that we can use $J_c$ to pull back the domain  $(\mathbb C\setminus[0,b], \mathbb H_\theta\setminus[0,b])$ to $\mathbb C\setminus(\gamma\cup[-1,0])$. In this way we can transform the RH problem for $(G,\widetilde G)$ to a scalar RH problem in the complex plane. Namely, let
\begin{equation}\label{def M}
M(s)\equiv
\begin{cases}
G(J_c(s))	&\text{ for $s$ outside }\gamma,\\
\widetilde{G}(J_c(s)) &\text{ for $s$ inside }\gamma.
\end{cases}
\end{equation}
Then, relying on Lemma \ref{lemma J}, we have that $M$ solves the Riemann-Hilbert problem
\subsubsection*{RH problem for $M$}
\bit
\item[(a)] $M$ is analytic in $\mathbb{C}\setminus\left(\gamma\cup\left[-1,0\right]\right)$.
\item[(b)] $M_+(s)+M_-(s) = V'(J_c(s))$ for $s\in\gamma\setminus\{-1,s_b\}$.\\
	$M_+(x)=e^{2\pi i/\theta}M_-(x)$ for $x\in (-1, 0)$, with $(-1,0)$ oriented from left to right.
\item[(c1)] $M(s) = \frac{1}{J_c(s)}\left(1+o(1)\right)$ as $s\to\infty$.
\item[(c2)]
	$M(s) = \frac{\theta}{J_c(s)}\left(1+o(1)\right)$ as $s\to 0$.
\eit
We can solve this RH problem explicitly. To that end, we let
\begin{equation}\label{def N}
N(s) = J_c(s)M(s).
\end{equation}
By the multiplication with $J_c(s)$, the jump along $(-1,0)$ is eliminated and the RH problem is transformed to the following form:
\subsubsection*{RH problem for $N$}
\bit
\item[(a)] $N$ is analytic in $\mathbb{C}\setminus\gamma$.
\item[(b)] $N_+(s) + N_-(s) = J_c(s)V'(J_c(s)) = U_c(s)$ for $s\in\gamma\setminus\{-1,s_b\}$.
\item[(c1)] $N(s)\to 1$ as $s\to\infty$.
\item[(c2)]$N(0)=\theta$.
\eit

A solution to the RH solution $N$ can be constructed immediately as follows. The function defined by
\begin{equation}\label{sol N}
N_0(s) = \begin{cases}
-\frac{1}{2\pi i}\oint_\gamma\frac{U_c(\xi)}{\xi-s}d\xi +1	&\text{outside }\gamma\\
\frac{1}{2\pi i }\oint_\gamma\frac{U_c(\xi)}{\xi-s}d\xi -1	&\text{inside }\gamma
\end{cases}
\end{equation}
satisfies conditions (a), (b), and (c1).
In order to have condition (c2) as well, we need $c$ to satisfy \eqref{getc}. We also have (\ref{b}) since $b=J_c(s_b)$.

\medskip

There is, however, a caveat: the solution to the RH problem for $N$ is not unique without imposing additional conditions about the behavior of $N$ near $-1$ and $s_b$, and therefore we cannot be sure that the function $N$ constructed by (\ref{g})-(\ref{tilde g}), (\ref{G}), (\ref{def M}), and (\ref{def N}) is the same as the function $N_0$ defined in (\ref{sol N}).
If we assume for a moment that $N=N_0$, we can reconstruct the density $\psi$ of $\mu$ by a direct computation. Indeed, by (\ref{g}) and (\ref{G}), for $x\in[0, b]$ we have
\begin{align}
\psi(x) 	&=-\frac{1}{2\pi i}(G_+(x)-G_-(x))=-\frac{1}{2\pi i x}(N_{-}(I_+(x))-N_{-}(I_-(x)))\label{psi4}\\
		&=-\frac{1}{4\pi^2 x}\oint_\gamma U_c(t)\left(\frac{1}{t-I_+(x)}-\frac{1}{t-I_-(x)}\right)dt\nonumber\\
		&=\frac{1}{2\pi^2 x}\int_0^bV'(y)y\;\text{Re}\left[\frac{I_+'(y)}{I_+(y)-I_+(x)}-\frac{I_+'(y)}{I_+(y)-I_-(x)}\right]dy\nonumber\\
		&=\frac{1}{2\pi^2x}\int_0^bV'(y)y\;\frac{d}{dy}\text{Re}\left[\log\frac{I_+(y)-I_+(x)}{I_+(y)-I_-(x)}\right]dy\nonumber\\
		&=\frac{1}{2\pi^2x}\int_0^b\Big(V''(y)y+V'(y)\Big)\log\left|\frac{I_+(y)-I_-(x)}{I_+(y)-I_+(x)}\right|dy,\label{psi3}
\end{align}
which is formula \eqref{psi}. In order to complete the proof of Theorem \ref{theorem: eq measure formula}, we only need to show that $N=N_0$.

\medskip

By the RH conditions for $N$ and $N_0$, we have that $Q(s)\equiv N(s)-N_0(s)$ has no jumps and is thus a meromorphic function, which tends to $0$ as $s\to\infty$ and which possibly has isolated singularities at $-1$ and $s_b$. If we show that these singularities are removable, we have $Q\equiv 0$ and $N\equiv N_0$ by Liouville's theorem.
By the definition of $N$ outside $\gamma$ (see (\ref{g}), (\ref{G}), (\ref{def M}), and (\ref{def N})), one verifies that $N(s)$ is bounded for $s$ outside $\gamma$ and near $-1$. Together with (\ref{sol N}), this implies that $-1$ is either a removable or an essential singularity of $Q$, the latter being impossible by the definition of $N$ inside $\gamma$ (see (\ref{tilde g}), (\ref{G}), (\ref{def M}), and (\ref{def N})).
For $s$ near $s_b$, again by the definition of $N$, $Q$ can only have a simple pole or a removable singularity at $s_b$. If it has a simple pole, we would have that $\psi(x)\sim c(b-x)^{-1/2}$ as $x\searrow b$, but one can show that this is in contradiction with the variational conditions (\ref{EL1})-(\ref{EL2}). We have that $N=N_0$, and this completes the proof of Theorem \ref{theorem: eq measure formula}.

\subsubsection{Proof of Theorem \ref{theorem: one-cut reg}}
In this section, we assume the conditions of Theorem \ref{theorem: one-cut reg} are satisfied. Then, without assuming one-cut regularity, we can pursue the construction done in the previous section, and define the density $\psi$ as before, with $c$ defined by (\ref{getc}) and $b$ by (\ref{b}).
By construction the measure $d\mu=\psi(x)dx$ satisfies the Euler-Lagrange condition (\ref{EL1}). In order to prove that our candidate-equilibrium measure is indeed the one, we need to prove that it is a probability measure, and that it satisfies the Euler-Lagrange inequality (\ref{EL2}).

Positivity of $\psi$ follows immediately from \eqref{psi3} and the assumption that $xV''(x)+V'(x)> 0$, since the term inside the logarithm of (\ref{psi3}) can be checked geometrically to be greater than $1$. The fact that $\int\psi(x)dx=1$ is equivalent to condition (c1) of the RH problem for $(G,\widetilde G)$ and this is true by construction.

To show the variational inequality outside the support, consider the function
$$
h(x) = G_+(x)+\widetilde{G}_-(x) - V'(x)
$$
We have $h(x)=0$ for $x\in[0, b]$ by condition (b) of the RH problem for $(G,\widetilde G)$. Now, by the assumption $V''(x)\geq 0$, for $x>b$ we have
\begin{align*}
\frac{d}{dx}h(x) &=\frac{d}{dx}(G(x)+\widetilde{G}(x)-V'(x))\\
			&=g''(x)+\widetilde{g}''(x)-V''(x)\\
			&=-\int_0^b\left(\frac{1}{(x-y)^2}+\frac{\theta x^{2\theta-2}+\theta(\theta-1)x^{\theta-2}y^\theta}{(x^\theta- y^\theta)^2}\right)\psi(y)dy-V''(x)<0
\end{align*}
where we used the fact that $\theta>1$ in the last inequality. We conclude that $h(x)<0$ for $x>b$, which yields (\ref{EL2}) after integration.\\
This shows that the density $\psi$ we constructed is indeed the equilibrium measure, and this completes the proof.

\subsection{The mapping $\widetilde J$}

Consider the transformation $\widetilde J$ defined in (\ref{def tilde J}). It has two critical points $s_a$ and $s_b$ given by (\ref{sab}), which are mapped to $a\equiv \widetilde J(s_a)$ and $b\equiv \widetilde J(s_b)$.
Write $s=re^{i\phi}$ with $0\leq \phi <2\pi$.
It is easy to verify that $\arg \widetilde J(s)=0$ if and only if
\begin{equation}
\arg\left(\frac{c_0}{c_1} +re^{i\phi}\right)+\frac{1}{\theta}\arg(1+re^{i\phi})=\frac{\phi}{\theta}.\label{arg tilde J}
\end{equation}
If $\frac{c_0}{c_1}>0$, given $0<\phi<\pi$, the left hand side is increasing in $r$, tends to $0$ as $r\to 0$, and to $\left(1+\frac{1}{\theta}\right)\phi$ as $r\to\infty$. We can conclude that there is a unique $r=r(\phi)$ such that (\ref{arg tilde J}) holds, or in other words such that $\widetilde J(re^{i\phi})>0$. We have that there exists a curve $\widetilde\gamma_1$ connecting $s_a$ with $s_b$ in the upper half plane which is mapped bijectively to $[a,b]$. By symmetry the complex conjugate curve $\widetilde\gamma_2$ is also mapped bijectively to $[a,b]$. The union $\widetilde\gamma$ of $\widetilde\gamma_1$ and $\widetilde\gamma_2$ is shown in Figure \ref{figure: J2} for different values of $\theta$ and $c_0/c_1$.
In a similar way as we did before for $J$, see Section \ref{section: J}, we can prove the following lemma.
\begin{figure}[t]
\begin{center}
\includegraphics[scale=0.38,angle=0]{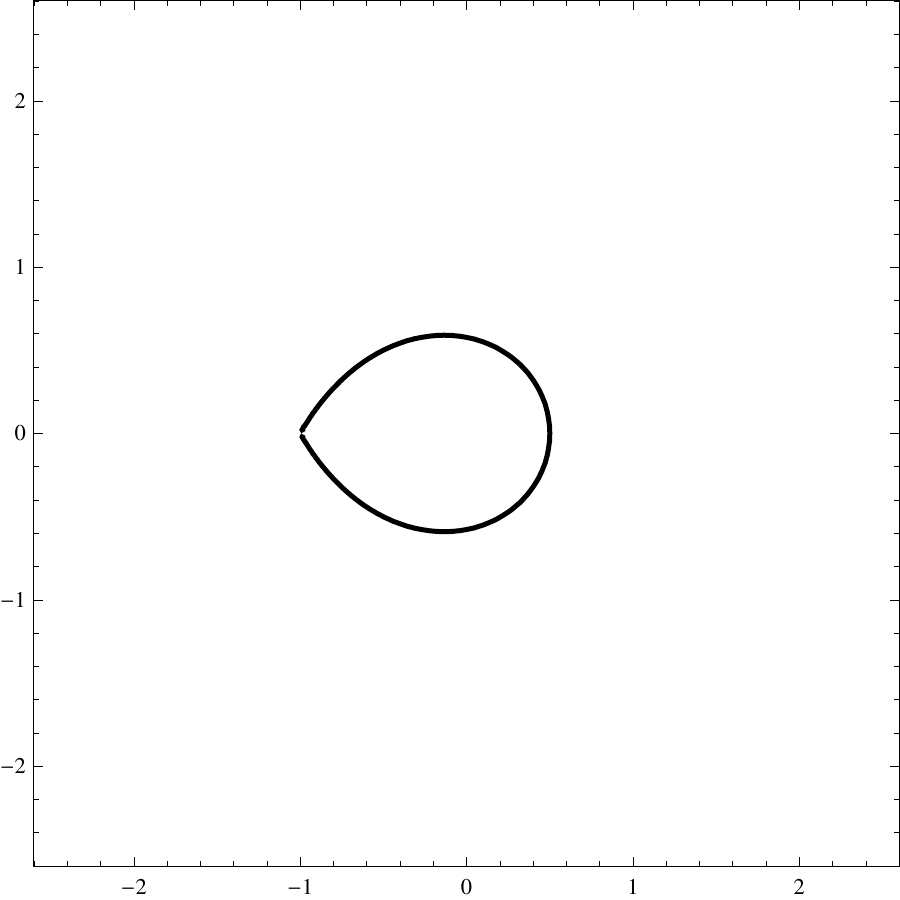}
\includegraphics[scale=0.38,angle=0]{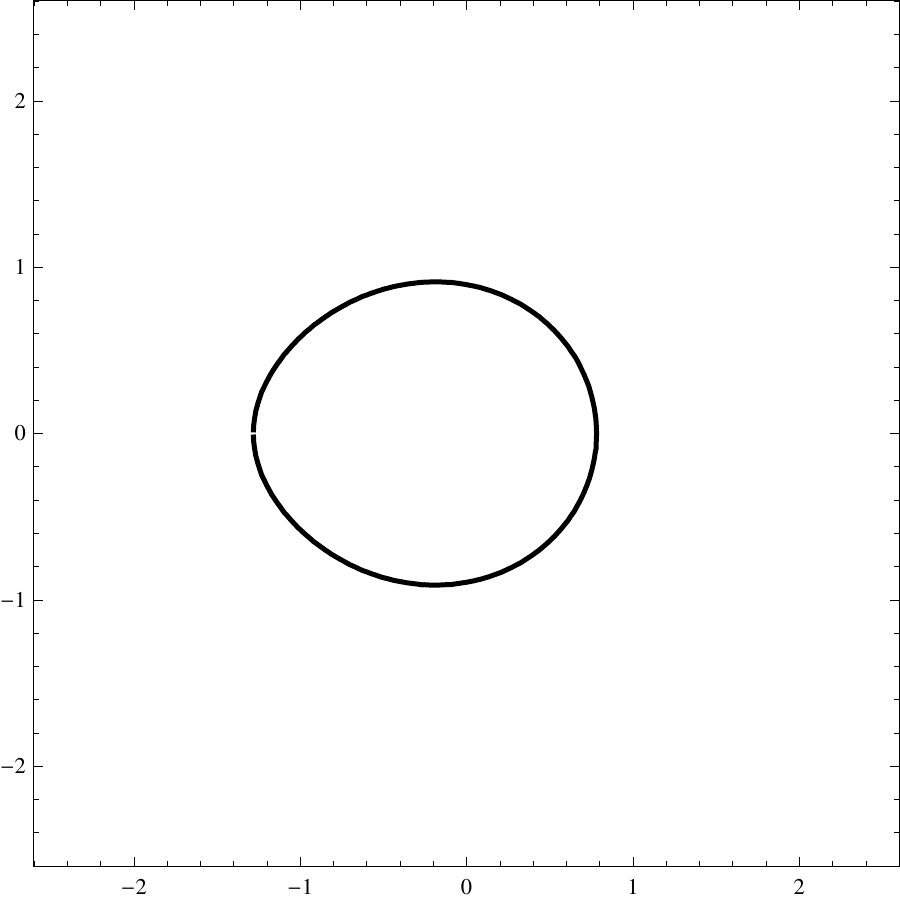}
\includegraphics[scale=0.38,angle=0]{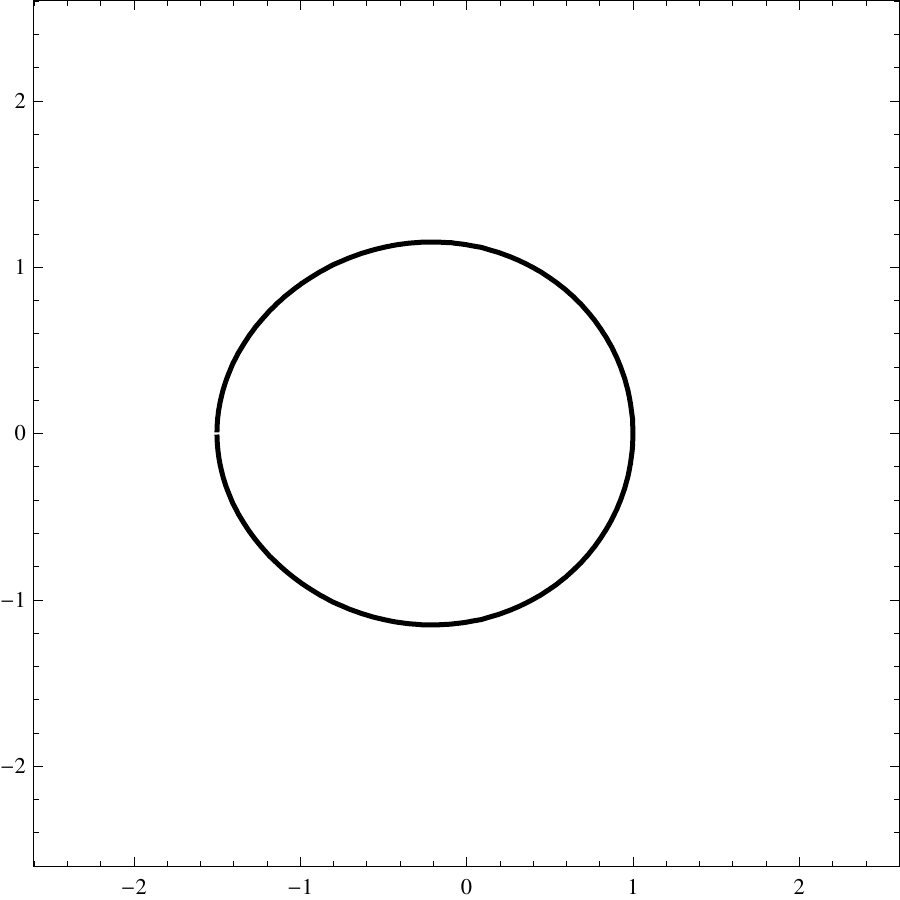}
\includegraphics[scale=0.38,angle=0]{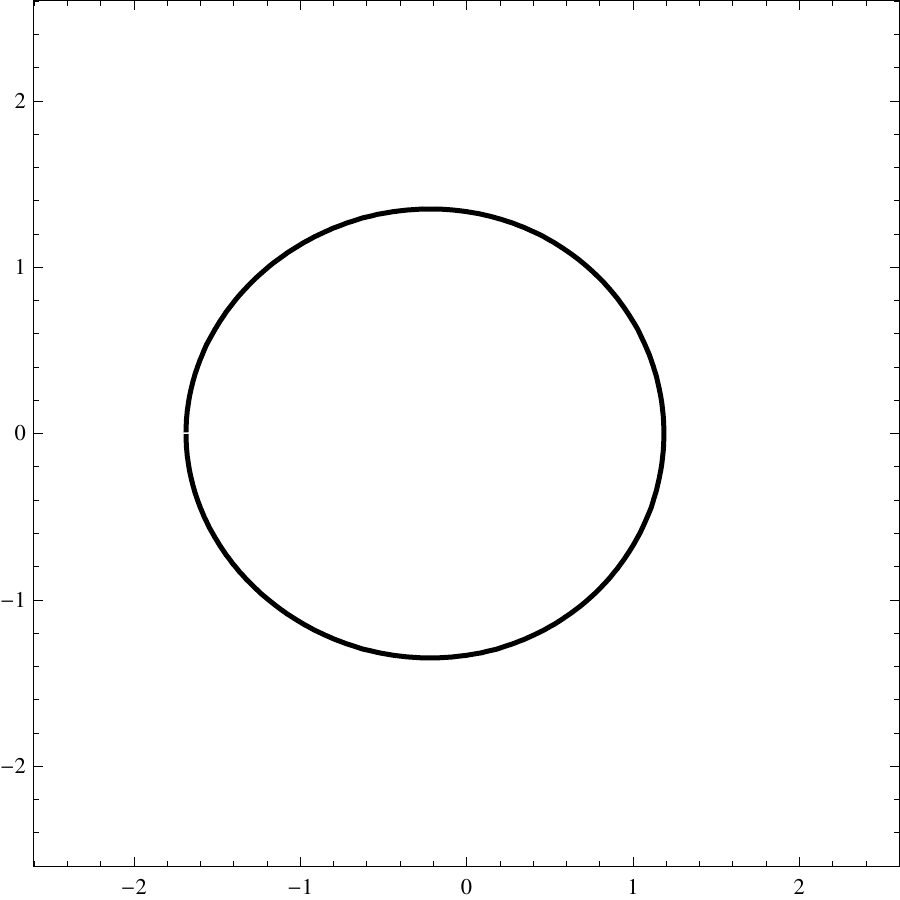}
\includegraphics[scale=0.38,angle=0]{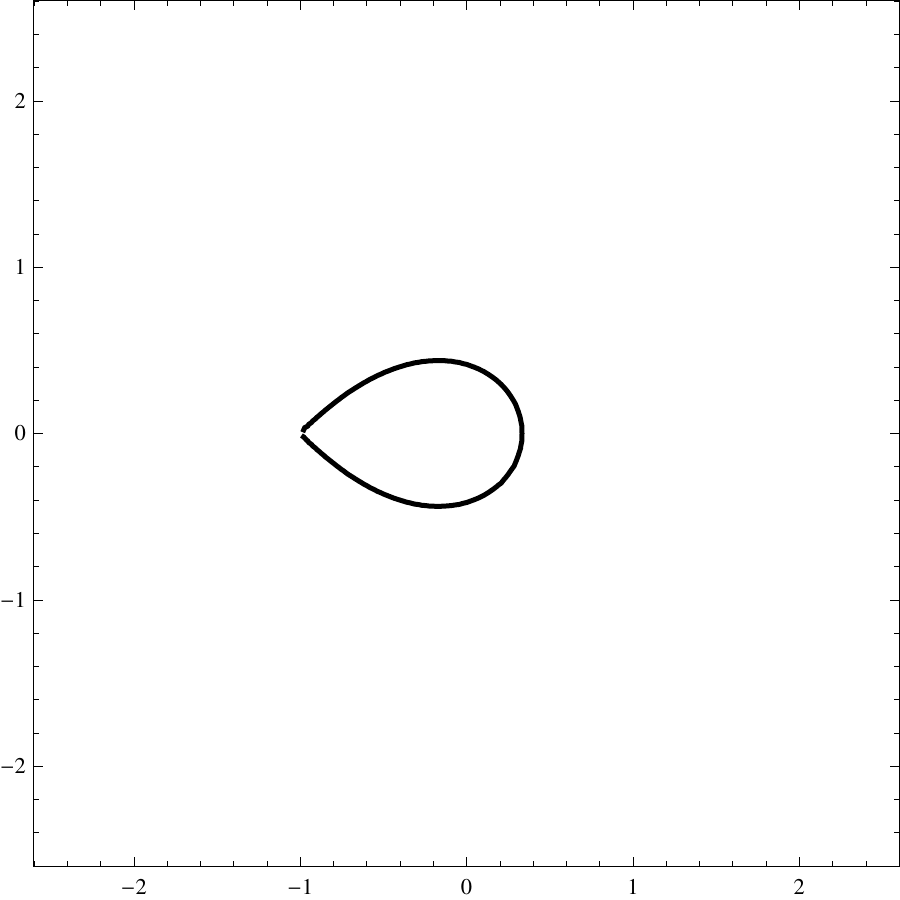}
\includegraphics[scale=0.38,angle=0]{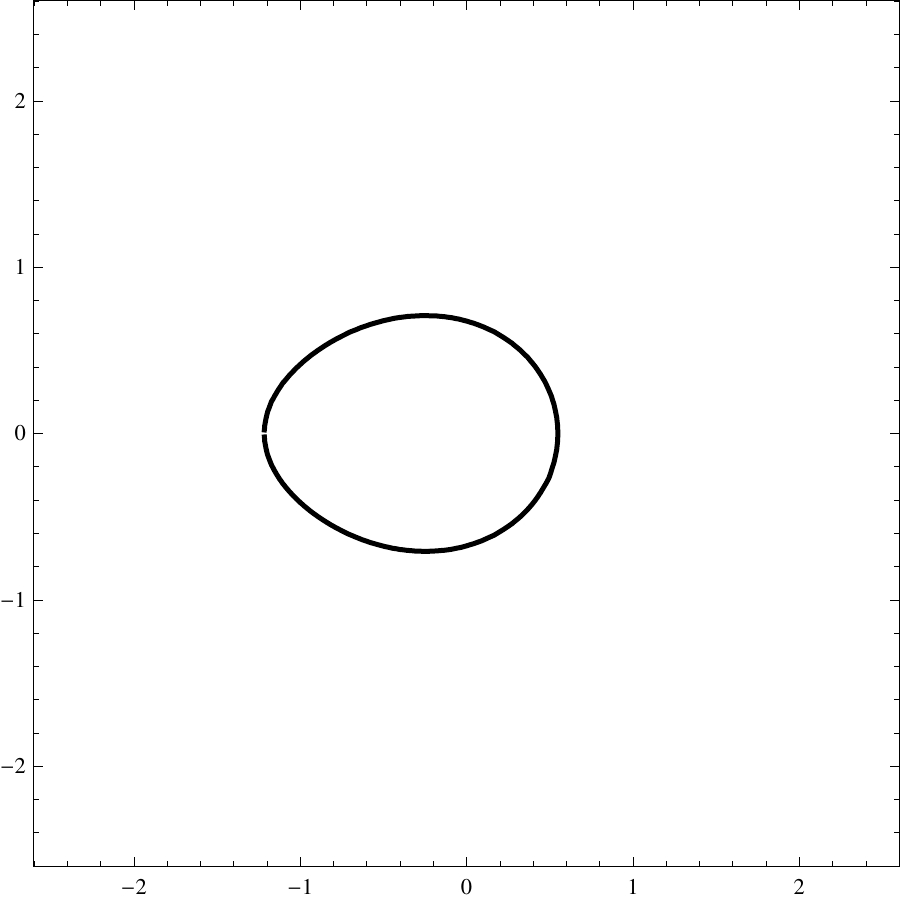}
\includegraphics[scale=0.38,angle=0]{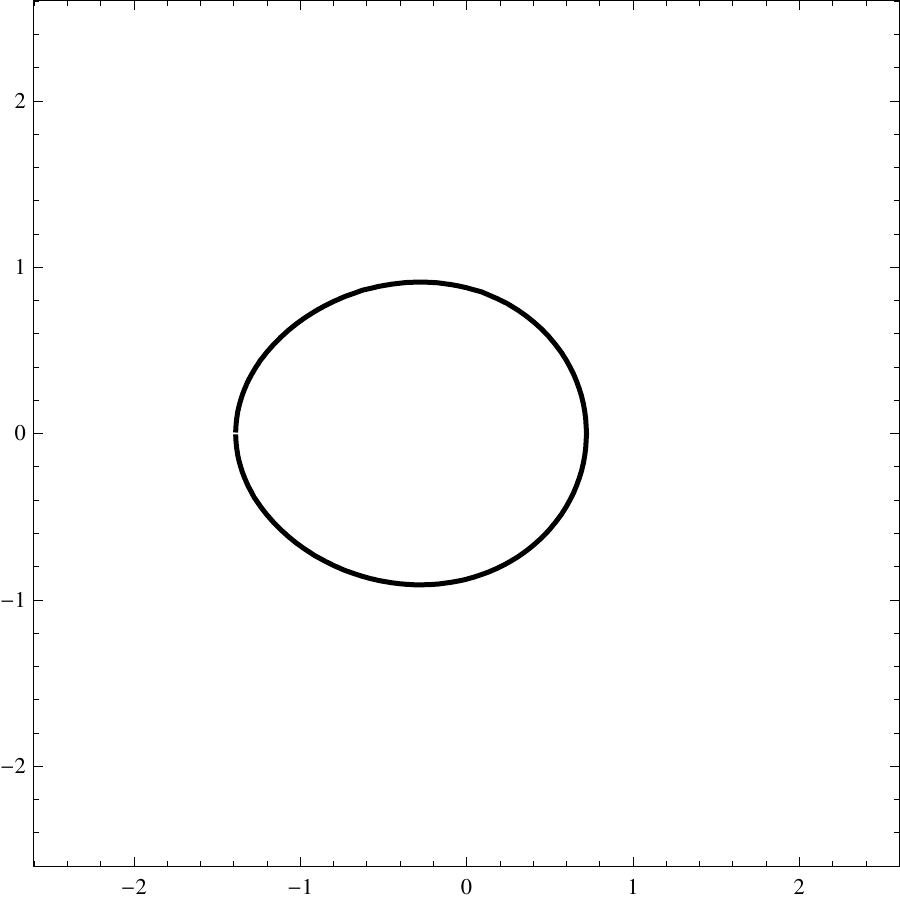}
\includegraphics[scale=0.38,angle=0]{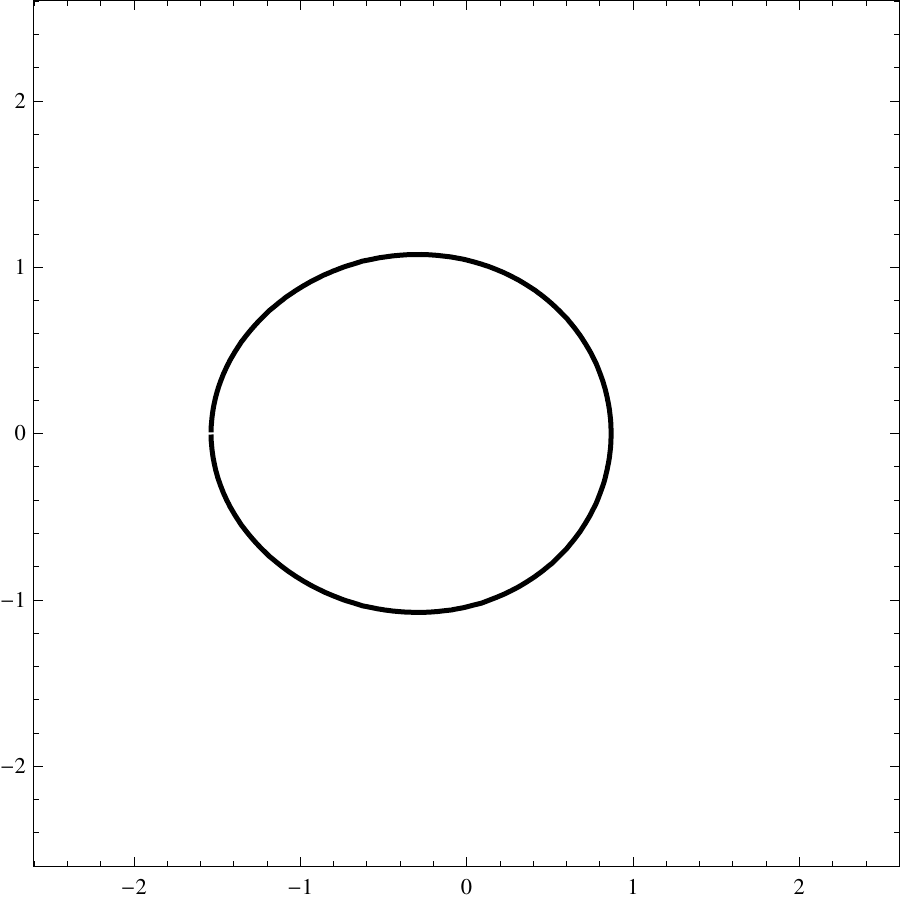}
\includegraphics[scale=0.38,angle=0]{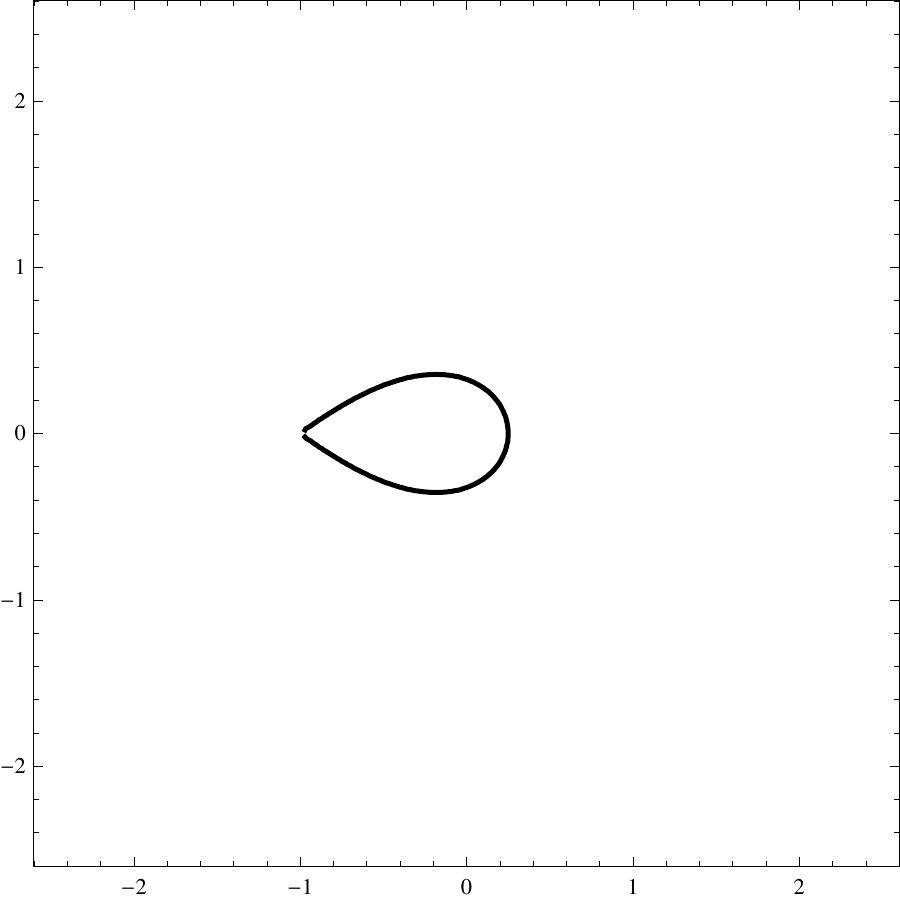}
\includegraphics[scale=0.38,angle=0]{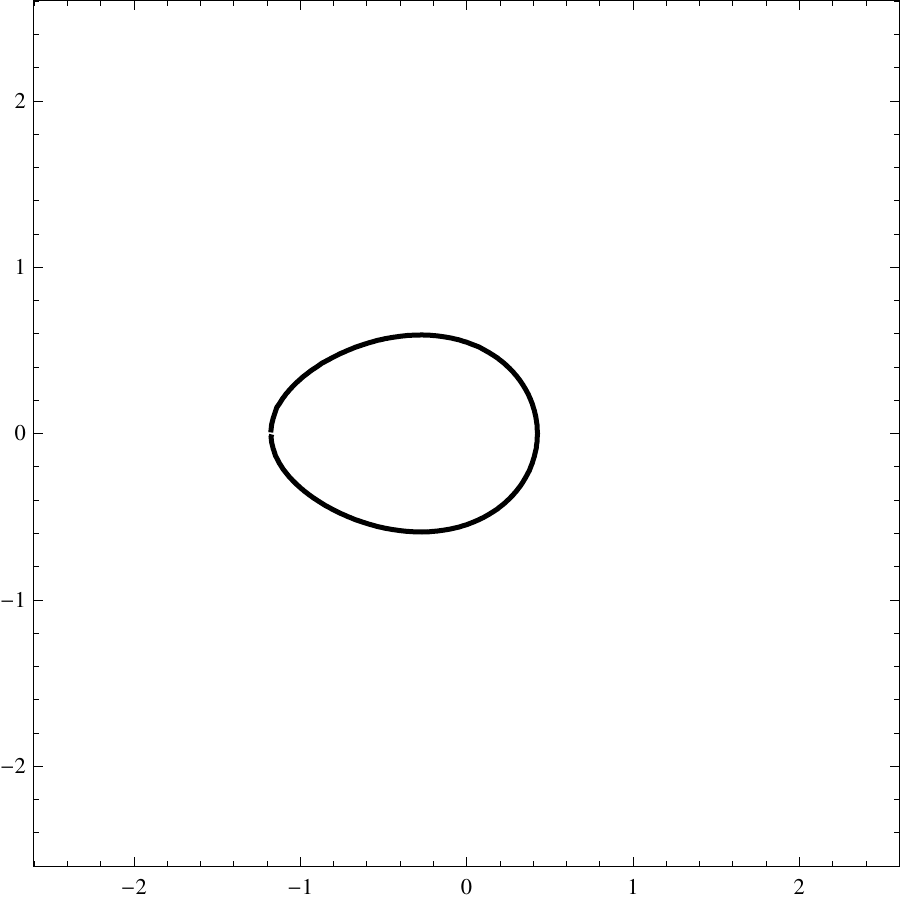}
\includegraphics[scale=0.38,angle=0]{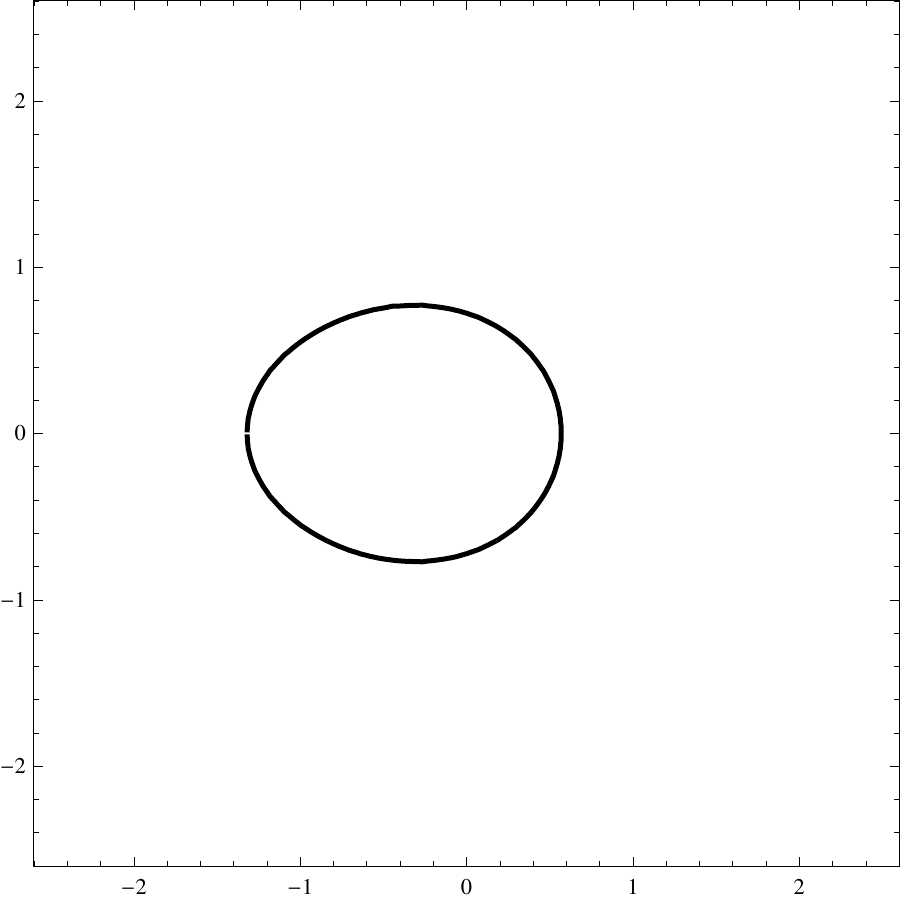}
\includegraphics[scale=0.38,angle=0]{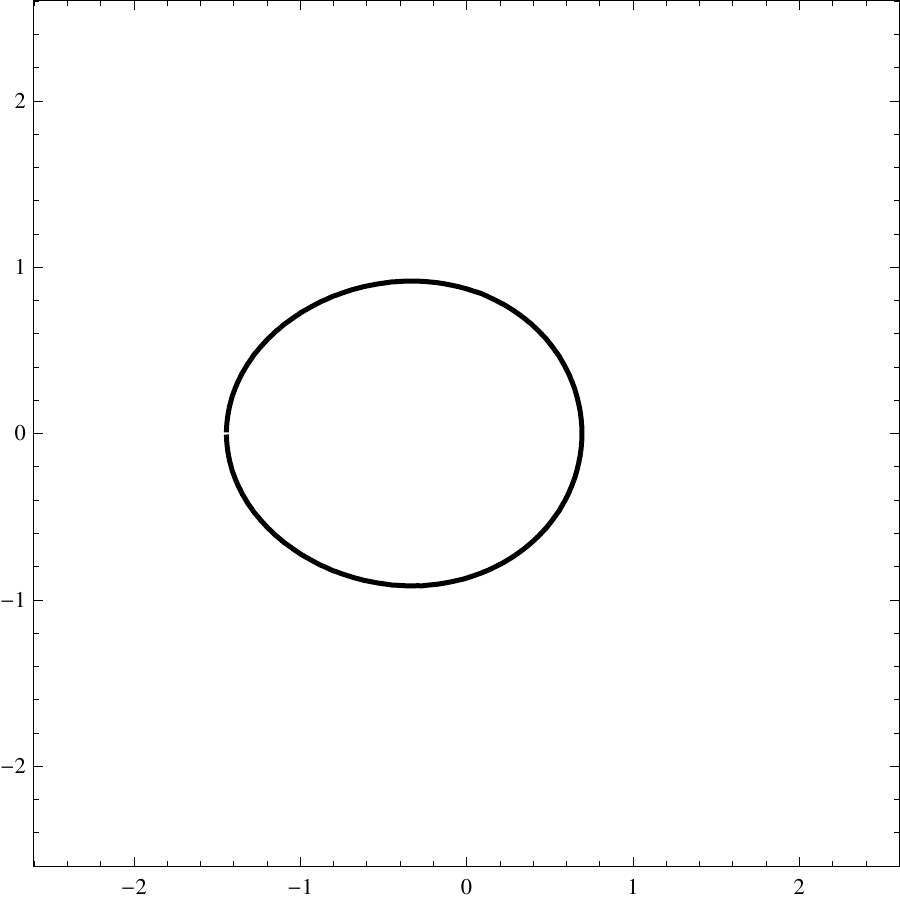}
\end{center}
\caption{The curve $\widetilde\gamma$ for $\theta=2$ and $\frac{c_0}{c_1}=1, 2, 3, 4$ (upper row), for $\theta=3$ and $\frac{c_0}{c_1}=1, 2, 3, 4$ (middle row), and for $\theta=4$ and $\frac{c_0}{c_1}=1, 2, 3, 4$ (bottom row).}
\label{figure: J2}
\end{figure}
\begin{lemma}
Let $c_0>c_1>0$ and let $D$ be the open region enclosed by the curves $\widetilde\gamma_1$ and $\widetilde\gamma_2$. $\widetilde J$ maps $\mathbb C\setminus \overline D$ bijectively to $\mathbb C\setminus [a,b]$, and it maps $D\setminus[-1,0]$ bijectively to $\mathbb H_\theta\setminus[a,b]$.
\end{lemma}

\subsection{Solution in the one-cut regular case without hard edge}\label{section: soft}
We assume that $\mu$ is a probability measure with a continuous density $\psi$ supported on $[a,b]$ which satisfies the Euler-Lagrange condition (\ref{EL1}).
Consider as before the logarithmic transforms
\begin{align}
g(z)&=\int_a^b\log(z-y)d\mu(y), &z\in\mathbb{C}\setminus(-\infty,b]\nn,\\
\widetilde{g}(z)&=\int_a^b\log(z^\theta-y^\theta)d\mu(y), &z\in\mathbb{H_\theta}\setminus[0,b].
\end{align} The variational condition (\ref{EL1}) translates into the following conditions for $G(z)=g'(z), \widetilde{G}(z)=\widetilde{g}'(z)$:
\subsubsection*{RH problem for $(G,\widetilde G)$}
\bit
\item[(a)] $(G(z), \widetilde{G}(z))$ is analytic in $(\mathbb{C}\setminus[a, b], \mathbb{H}_\theta\setminus[a, b])$,
\item[(b)] $G_{\pm}(x)+\widetilde{G}_{\mp}(x)=V'(x)$ for $x\in (a, b)$,\\
$\widetilde{G}(e^{-\pi i/\theta}x)=e^{2\pi i/\theta}\widetilde{G}(e^{\pi i/\theta}x)$ for $x>0$.
\item[(c1)] $G(z)\sim \frac{1}{z}$ as $z\to\infty$,
\item[(c2)]$\widetilde{G}(z)\sim \frac{\theta}{z}$ as $z\to\infty$ in $\mathbb{H}_\theta$,
\eit
We proceed in a similar way as in the hard edge case, and pull back the domain $(\mathbb{C}\setminus[a, b], \mathbb{H}_\theta\setminus[a, b])$ to $\mathbb C\setminus(\widetilde\gamma\cup[-1,0])$ using $\widetilde J$, where $\widetilde\gamma$ is the counterclockwise oriented union of $\widetilde\gamma_1$ and $\widetilde\gamma_2$.
Introducing
\beq
N(s)=\begin{cases}
\widetilde J_{c_0, c_1}(s)G(\widetilde J_{c_0, c_1}(s))\;\;\; \text{outside } \widetilde\gamma\\
\widetilde J_{c_0, c_1}(s)\widetilde{G}(\widetilde J_{c_0, c_1}(s))\;\;\; \text{inside } \widetilde\gamma
\end{cases}
\eeq
we obtain the following RH problem for $N$:
\subsubsection*{RH problem for $N$}
\bit
\item[(a)] $N$ analytic in $\mathbb{C}\setminus \widetilde\gamma$.
\item[(b)] $N_+(s)+N_-(s)=U_{c_0, c_1}(s)$ for $s\in\widetilde\gamma\setminus\{s_a,s_b\}$,
where \[U_{c_0,c_1}(s)=V'(\widetilde J_{c_0,c_1}(s))\widetilde J_{c_0,c_1}(s).\]
\item[(c1)] $N(s)\to 1$ as $s\to\infty$.
\item[(c2)] $N(0)=\theta$, $N(-1)=0$.
\eit
As in the hard edge case, we can again show that if $\mu$ satisfies the variational conditions (\ref{EL1})-(\ref{EL2}), $N$ has to be bounded near $s_a$ and $s_b$.
The unique solution to conditions (a), (b), and (c1) which is bounded near $s_a$ and $s_b$ is given by
\beq
N(s)=\begin{cases}
-\frac{1}{2\pi i}\oint_{\widetilde\gamma}\frac{U_{c_0, c_1}(\xi)}{\xi-s}d\xi+1\;\;\;\text{outside of }\widetilde\gamma\\
\frac{1}{2\pi i}\oint_{\widetilde\gamma}\frac{U_{c_0, c_1}(\xi)}{\xi-s}d\xi-1\hspace{20pt}\text{inside of }\widetilde\gamma.\\
\end{cases}
\label{solution N}\eeq
In order to have condition (c2), we need that $c_0,c_1$ satisfy the equations (\ref{eqc0})-(\ref{eqc1}).
Then we have
\begin{align*}
\psi(x) 	&=-\frac{1}{2\pi i}(G_+(x)-G_-(x))=-\frac{1}{2\pi i x}(N_-(\widetilde I_+(x))-N_-(\widetilde I_-(x)))\\
		&=-\frac{1}{4\pi^2 x}\oint_{\widetilde\gamma} U_{c_0, c_1}(t)\left(\frac{1}{t-\widetilde I_+(x)}-\frac{1}{t-\widetilde I_-(x)}\right)dt\\
		&=\frac{1}{2\pi^2 x}\int_a^bV'(y)y\;\text{Re}\left[\frac{1}{\widetilde I_+(y)-\widetilde I_+(x)}-\frac{1}{\widetilde I_+(y)-\widetilde I_-(x)}\right]\widetilde{I}_+'(y)dy\\
		&=\frac{1}{2\pi^2x}\int_a^bV'(y)y\;\frac{d}{dy}\text{Re}\left[\log\frac{\widetilde I_+(y)-\widetilde I_+(x)}{\widetilde I_+(y)-\widetilde I_-(x)}\right]dy\\
		&=\frac{1}{2\pi^2x}\int_a^b(V''(u)u+V'(u))\log\left|\frac{\widetilde I_+(y)-\widetilde I_-(x)}{\widetilde I_+(y)-\widetilde I_+(x)}\right|dy.
\end{align*}
This gives us formula (\ref{psi2}), and Theorem \ref{theorem: eq measure formula soft edge} is proved.

\subsection{Examples}\label{section: ex}

If $\theta=2$, we have $J_c(s)=c\left(s+1\right)^{3/2}s^{-1/2}$. In order to compute the inverses $I_\pm(x)$, we need to find two complex conjugate roots of the equation $J_c(s)=x$, which reduces to the third degree equation
\[s^3+3s^2+\left(3-\frac{x^2}{c^2}\right)s+1=0.\]
By Cardano's formulas, we have
\begin{equation}\label{Ipm2}
I_\pm(x)=\frac{x^{2/3}}{2^{1/3}c^{2/3}}\left(u\chi_\mp(t)+\overline{u}\chi_\pm(t)\right)-1,
\end{equation}
where
\begin{equation}\label{chi u}
\chi_\pm(t)=\left(1\pm\sqrt{1-\frac{4x^2}{27c^2}}\right)^{1/3},\qquad u=\frac{1+i\sqrt{3}}{2}.
\end{equation}
Substituting (\ref{Ipm2}) and (\ref{chi u}) into (\ref{psi}), one has an explicit integral formula for the equilibrium density $\psi$.

\subsubsection{The Laguerre case}
\begin{figure}[t]
\begin{center}
\includegraphics[scale=0.7,angle=0]{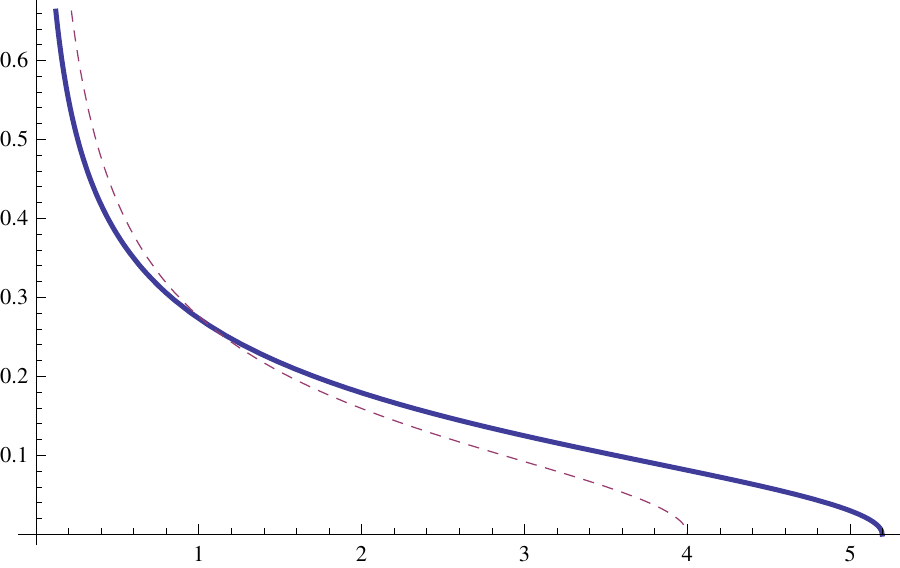}
\end{center}
\caption{The equilibrium density $\psi(x)$ for the Laguerre weight $e^{-nx}$ if $\theta=2$ (solid curve) and $\theta=1$ (dashed curve).  \label{figure: Lag2}}
\end{figure}

We let $V(x)=\rho x$ and $\theta>1$. The conditions of Theorem \ref{theorem: one-cut reg} are satisfied. By (\ref{getc0}),
\begin{equation}\label{c-ex2}
c=2\pi i\frac{1+\theta}{\rho}\left(\int_\gamma \left(\frac{s+1}{s}\right)^{1+1/\theta}ds\right)^{-1}=\frac{\theta}{\rho},
\end{equation}
and by (\ref{b}), the endpoint $b$ is given by $b=\frac{1}{\rho}(1+\theta)^{1+1/\theta}$.
The RH problem for $N$ can be solved explicitly in this case: we have
\begin{equation}\label{Nsol2}
N(s)=\begin{cases}\theta s+\theta,& \mbox{$s$ inside $\gamma$,}\\
\theta(s+1)\left(\frac{s+1}{s}\right)^{\frac{1}{\theta}}-\theta s-\theta,&\mbox{$s$ outside $\gamma$}.
\end{cases}
\end{equation}
Using (\ref{def N}) and (\ref{def M}), we obtain
\begin{equation}
\widetilde G_\pm(x)=\frac{1}{x}\left(\theta I_\mp(x)+\theta\right),\qquad 0<x<b,
\end{equation}
and by (\ref{tilde g}) and (\ref{G}), we have
\begin{eqnarray}
\psi(x)&=&-\frac{1}{2\pi i}\left(\widetilde G_+(x)-\widetilde G_-(x)\right)\nonumber\\
&=&\frac{\theta}{2\pi ix}\left(I_+(x)-I_-(x)\right)\label{densityLag}.
\end{eqnarray}
For $\theta=2$, we have
$b=\frac{3\sqrt{3}}{\rho}$, and we can substitute the explicit expressions (\ref{Ipm2})-(\ref{chi u}) for $I_\pm$ into (\ref{densityLag}).
We then obtain
\begin{equation}
\psi(x)=\frac{\sqrt{3}\rho^{2/3}}{2\pi x^{1/3}}\left(\left(1+\sqrt{1-\frac{\rho^2x^2}{27}}\right)^{1/3}-\left(1-\sqrt{1-\frac{\rho^2x^2}{27}}\right)^{1/3}\right).\end{equation}
This formula for the density, which is shown in Figure \ref{figure: Lag2}, is the same as the one obtained in \cite{LSZ}.

\subsubsection{$V(x)=\tau x^2+\rho x$: hard-to-soft edge transition}\label{section: hardsoft}
\begin{figure}[t]
\begin{center}
\includegraphics[scale=0.5,angle=0]{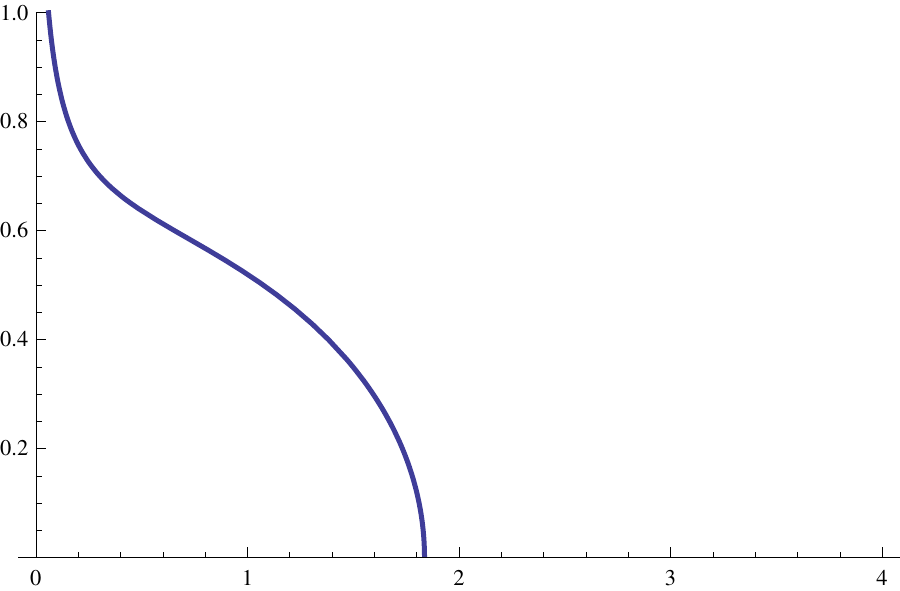}
\includegraphics[scale=0.5,angle=0]{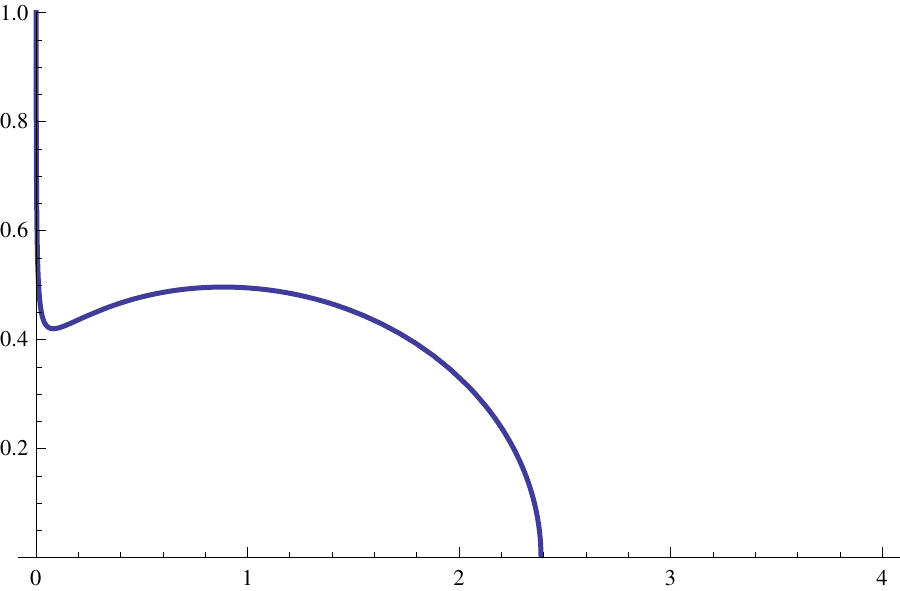}
\includegraphics[scale=0.5,angle=0]{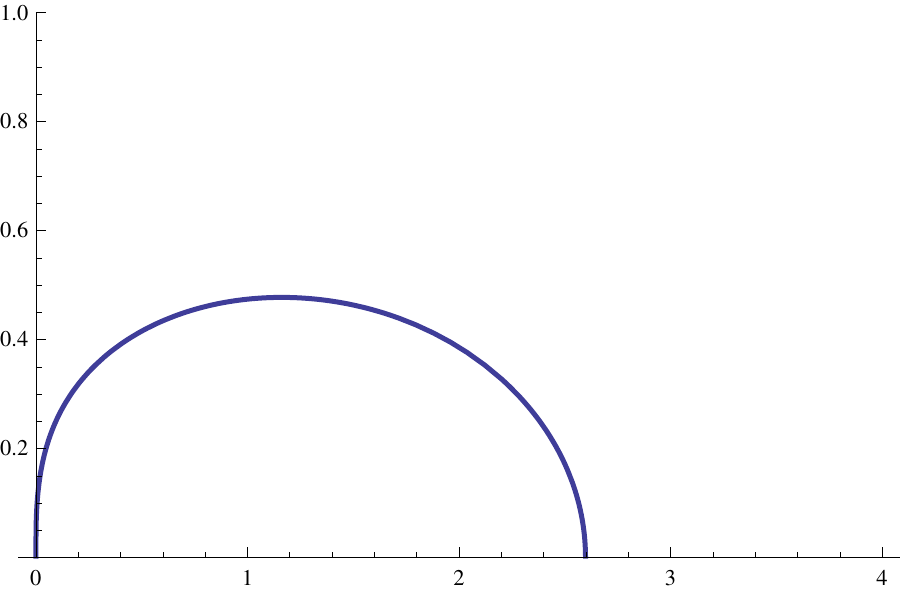}
\includegraphics[scale=0.5,angle=0]{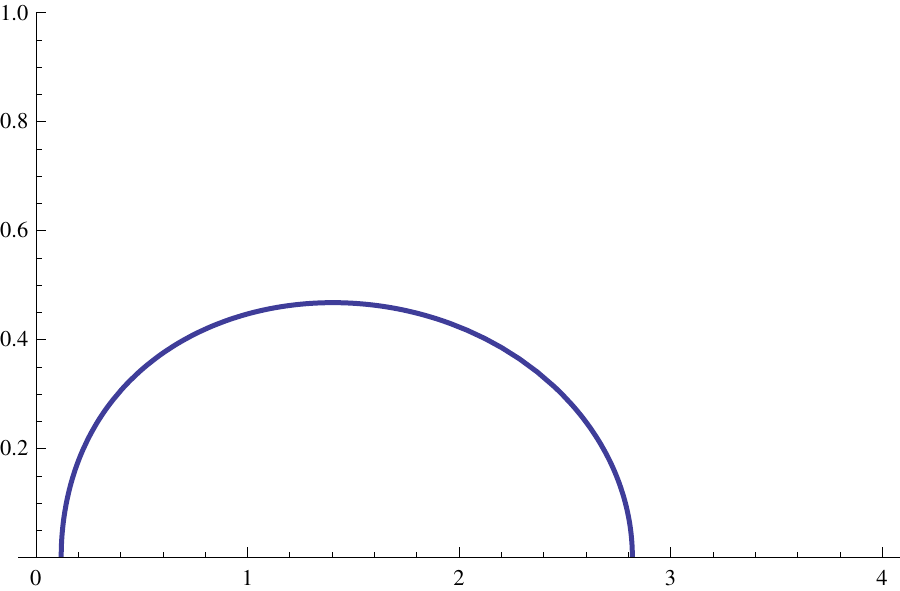}
\includegraphics[scale=0.5,angle=0]{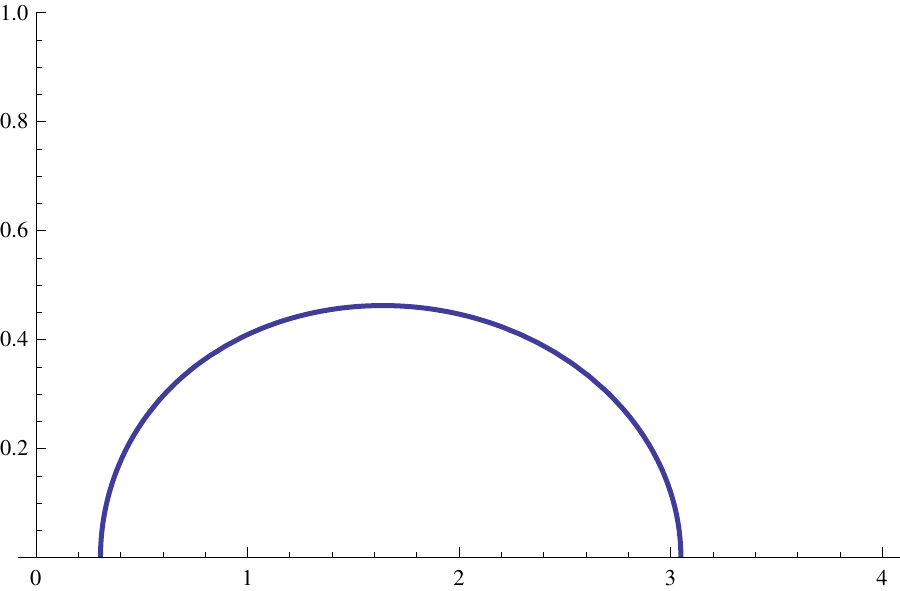}
\includegraphics[scale=0.5,angle=0]{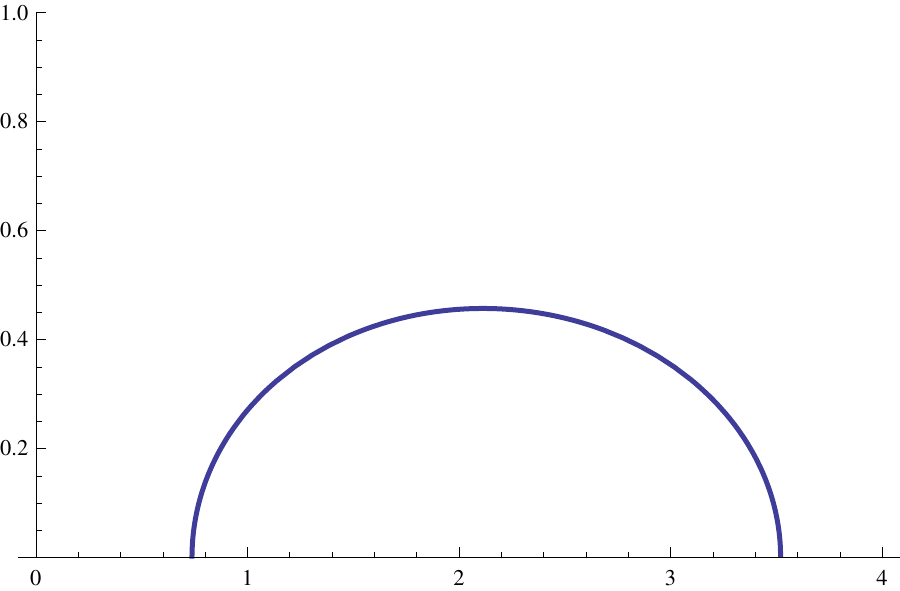}
\includegraphics[scale=0.5,angle=0]{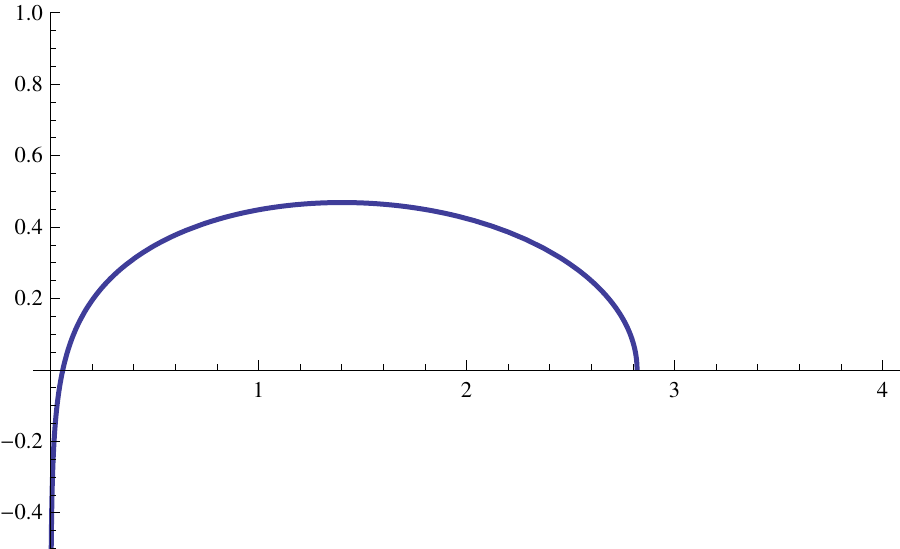}
\includegraphics[scale=0.5,angle=0]{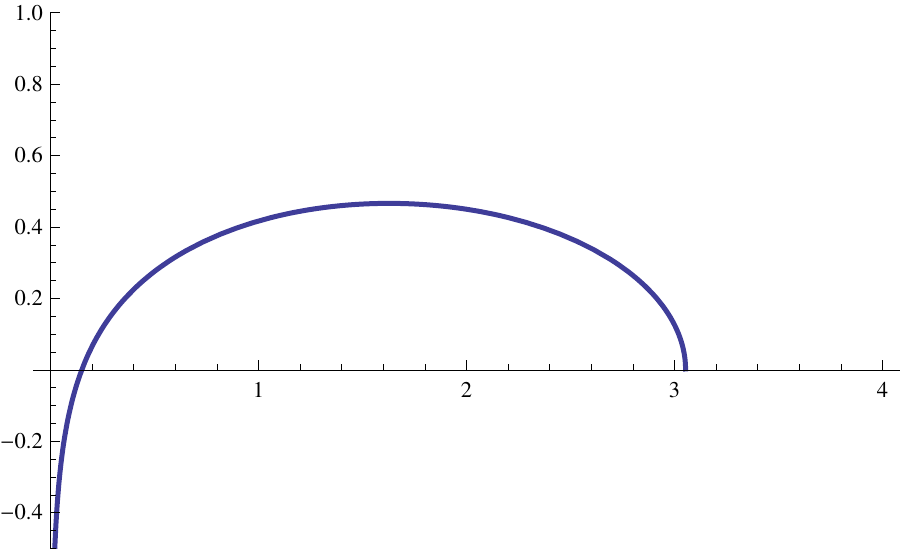}
\includegraphics[scale=0.5,angle=0]{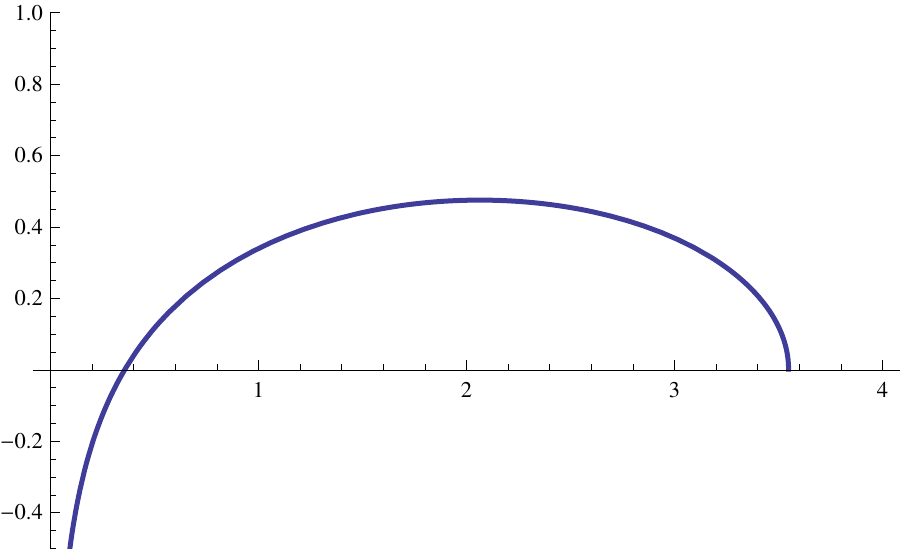}
\end{center}
\caption{The equilibrium density $\psi(x)$ for $V(x)=x^2+\rho x$ if $\rho=0$, $\rho=-1.8$, and $\rho=-2$ (upper row), $\rho=-2.5$, $\rho=-3$, and $\rho=-4$ (middle row), with $\theta=2$. Observe the transition where the hard edge turns into a soft edge, at $\rho=-2$.
For comparison, in the bottom row, the densities constructed under the (false) one-cut assumption with hard edge for $\rho=-2.5$, $\rho=-3$, and $\rho=-4$, which are negative near $0$.}
\label{figure: hardsoft}
\end{figure}

We let $V(x)=\tau x^2+\rho x$, with $\tau>0$ and $\rho\in\mathbb R$. Then $V''(x)=2\tau>0$, and $xV''(x)+V'(x)=4\tau x+\rho$. This means that the conditions given in Theorem \ref{theorem: one-cut reg} to have a one-cut supported equilibrium measure with hard edge are satisfied only if $\rho\geq 0$.
Nevertheless, even for $\rho<0$ , we can pursue with the construction done in Section \ref{section: eq hard} under the one-cut assumption with hard edge, without knowing that the constructed measure will be the equilibrium measure satisfying (\ref{EL1})-(\ref{EL2}).

By residue calculus, it follows from (\ref{getc0}) that $c$ is given by
\begin{equation}
c=\label{c-ex4}\frac{-\rho \theta+\theta \sqrt{\rho^2+16 \tau+8 \tau \theta}}{4 (2 \tau+\tau \theta)},
\end{equation}
and
\begin{equation}
b=\frac{-\rho + \sqrt{\rho^2+16 \tau+8 \tau \theta}}{4 (2 \tau+\tau \theta)}(1+\theta)^{1+1/\theta}.
\end{equation}

The RH solution $N$ can be constructed explicitly, and it is given by \begin{equation*}\label{Nsol3}
N(s)=\begin{cases}2\tau c^2s^2+4\tau c^2\frac{\theta+1}{\theta} s +\rho c s + \rho c\frac{\theta+1}{\theta}+2\tau c^2\frac{(\theta+1)(\theta+2)}{\theta^2}-1\equiv N_{in}(s),& \mbox{$s$ inside $\gamma$,}\\
J_c(s)V'(J_c(s))-N_{in}(s),&\mbox{$s$ outside $\gamma$},
\end{cases}
\end{equation*}
where $N_{in}(s)$ denotes the analytic continuation of $N$ from the inside of $\gamma$.
This gives the formula
\begin{eqnarray}
\psi(x)&=&-\frac{1}{2\pi i}\left(\widetilde G_+(x)-\widetilde G_-(x)\right)\nonumber\\
&=&\frac{1}{2\pi ix}\left(N_{in}(I_+(x))-N_{in}(I_-(x))\right).\label{densityex3}\end{eqnarray}
Let us take a closer look at the case $\theta=2$. One can then substitute the explicit expressions (\ref{Ipm2}) for $I_\pm$.
The density $\psi(x)$ constructed in this way is shown in Figure \ref{figure: hardsoft} for $\tau=1$ for different values of $\rho$. We observe that the constructed density is a probability density for $\rho\geq -2$, but that it becomes negative near $0$ for $\rho<-2$.
It is natural to expect that the constructed density is the true equilibrium density for $\rho\geq -2$, although the conditions of Theorem \ref{theorem: one-cut reg} are not satisfied.
If $\rho<-2$, the one-cut assumption with hard edge is false. Following the construction done in Section \ref{section: soft}, we can construct the equilibrium density under the assumption that it is one-cut supported without hard edge.

Substituting (\ref{def tilde J}) for $\theta = 2$ and $V(x) = x^2 + \rho x$ into (\ref{eqc0})-(\ref{eqc1}),  we find the solution
$$
c_0 = -\frac{\rho}{2},\qquad c_1 = -\frac{2}{\rho},
$$
which indeed satisfies $c_0 > c_1 > 0$ for all $\rho < -2$. Plugging this into (\ref{sab})-(\ref{sab2}) we get
$$
s_a = -\frac{1}{4}(1 + \sqrt{2\rho^2 + 1}), \qquad s_b = -\frac{1}{4}(1 - \sqrt{2\rho^2 + 1}),
$$
and $a=\tilde{J}(s_a)$, $b=\tilde{J}(s_b)$.
Again we can write the solution (\ref{solution N}) of the RH problem for $N$ explicitly. Inside $\tilde{\gamma}$, it is given by
$$
N_{in}(s) = \frac{2}{\rho^2}(4s + \rho^2)(s + 1).
$$
As before, the inverse of $\tilde{J}$ in the upper/lower half plane is given explicitly in terms of Cardano's formulas, and we can finally substitute in (\ref{densityex3}) to obtain the expression of the equilibrium density.
This construction leads indeed to a valid density, see Figure \ref{figure: hardsoft}, which  behaves locally like a square root at both of the endpoints $a$ and $b$.
As $\rho\searrow -2$, $a$ approaches $0$, and for $\rho=-2$, we have the local behavior
\begin{equation}
\psi(x)=Cx^{1/3}(1+o(1)),\qquad x\to 0.
\end{equation}
For general $\theta$, one expects a critical value $\rho_c=\rho_c(\theta)<0$ such that the equilibrium measure is one-cut supported without hard edge for $\rho<\rho_c$, one-cut supported with hard edge and with local behavior
\begin{equation}
\psi(x)=d_1x^{-\frac{1}{\theta+1}}(1+o(1)),\qquad x\to 0,
\end{equation}
for $\rho>\rho_c$, and such that
\begin{equation}
\psi(x)=Cx^{\frac{\theta-1}{\theta+1}}(1+o(1)),\qquad x\to 0,
\end{equation}
for $\rho=\rho_c$. At the critical value, the constant $d_1$ in (\ref{exponent})-(\ref{c1}) must vanish.

\section{Scaling limits of the correlation kernel}\label{section: univ}

In the previous sections, we explained how one can compute the equilibrium measure associated to the point processes (\ref{probability distribution}) if $w=e^{-nV}$, which is expected to describe the macroscopic behavior of the particles in the large $n$ limit. In order to have more detailed information about the microscopic behavior of the particles, one is typically interested in scaling limits of the correlation kernel (\ref{Kn}). Such scaling limits have been investigated in detail in the random matrix case $\theta=1$ and have lead to a number of universal limiting kernels if $w(x)=x^\alpha e^{-nV(x)}$.
Loosely speaking, for $\theta=1$, there are three different {\em regular} scaling limits:
\begin{itemize}
\item[(i)] If $x>0$ lies in the bulk of the spectrum, i.e.\ $x$ lies in the interior of the support of the equilibrium measure and the density $\psi(x)>0$, then the scaled correlation kernel tends to the sine kernel:
    \begin{equation}\label{sine}
    \lim_{n\to\infty}\frac{1}{\psi(x)n}K_n\left(x+\frac{u}{\psi(x)n},x+\frac{v}{\psi(x)n}\right)=\frac{\sin \pi(u-v)}{\pi(u-v)}.
    \end{equation}
\item[(ii)] If $x>0$ is a soft edge, i.e.\ an endpoint of the support of $\mu$ where the density $\psi$ behaves locally like a square root, the scaled correlation kernel tends to the Airy kernel:
    \begin{equation}\label{Airy}
    \lim_{n\to\infty}\frac{1}{|c|n^{2/3}}K_n\left(x+\frac{u}{cn^{2/3}},x+\frac{v}{cn^{2/3}}\right)=\frac{\Ai(u)\Ai'(v)-\Ai'(u)\Ai(v)}{u-v},
    \end{equation} for some suitably chosen $c\in\mathbb R$.
\item[(iii)] If the equilibrium density blows up like one over a square root at the hard edge $0$, the scaled correlation kernel tends to the Bessel kernel: if $u,v>0$, we have
    \begin{equation}\label{Bessel}
    \lim_{n\to\infty}\frac{1}{cn^{2}}K_n\left(\frac{u}{cn^{2}},\frac{v}{cn^{2}}\right)=\frac{J_\alpha(\sqrt{u})\sqrt{v}J_\alpha'(\sqrt{v})-\sqrt{u}J_\alpha'(\sqrt{u})J_\alpha(\sqrt{v})}{2(u-v)},
    \end{equation}
    for some $c>0$.
\end{itemize}
Apart from the regular scaling limits, various {\em singular} double scaling limits have been obtained if the equilibrium density behaves differently than in the three regular cases. One double scaling limit worth mentioning here, is one that occurs when the left endpoint $a$ of the support of the equilibrium measure approaches the hard edge $0$.
This leads to a family of kernels built out of Painlev\'e II functions \cite{ClaeysKuijlaars3}.

\medskip

Proving similar scaling limits for $\theta>1$ and general $V$ is out of reach for now, since we would need asymptotics for the biorthogonal polynomials for that. For weights of the form $w(x)=x^\alpha e^{-nV(x)}$, it is generally believed that scaling limits of the correlation kernel near a point $x^*$ are determined only by the local behavior of the equilibrium density $\psi(x)=\psi_{V,\theta}(x)$ near $x^*$, except if $x^*=0$, when the limiting kernel will also depend on the value of $\alpha$.
Away from zero, we have shown in this paper that the equilibrium density typically vanishes like a square root at soft edges of the support, and there is thus no difference in the local behavior of the density compared to the case $\theta=1$. Therefore, one can expect that the sine and Airy limits (i) and (ii) are still valid for general $\theta>1$. This is not true for the Bessel kernel near $0$.
In the Laguerre case where $w(x)=x^\alpha e^{-nx}$, Borodin \cite{Borodin} proved that
\begin{equation}
\lim_{n\to\infty}\frac{1}{n^{1+1/\theta}}K_n\left(\frac{u}{n^{1+1/\theta}},\frac{v}{n^{1+1/\theta}}\right)=\mathbb K_{\alpha,\theta}(u,v),
\end{equation}
where the limiting kernel $\mathbb K_{\alpha,\theta}$ is given by
\begin{equation}\label{gen Bessel}
K_{\alpha,\theta}(u,v)=\theta(uv)^{\alpha/2} \int_0^1 J_{\frac{\alpha+1}{\theta},\frac{1}{\theta}}(ut) J_{\alpha+1,\theta}((vt)^\theta)t^\alpha dt,
\end{equation}
and $J_{\alpha,\theta}$ is Wright's generalized Bessel function, also known as the Bessel-Maitland function, given by
\begin{equation}
J_{a,b}(x)=\sum_{m=0}^\infty \frac{(-x)^m}{m!\Gamma(a+bm)}.
\end{equation}
If $\theta=1$, $\mathbb K_{\alpha,1}$ is the well-known hard-edge Bessel kernel \cite{KuijlaarsVanlessen} given in scaling limit (iii) above.
In view of the universal nature of the scaling limits of the correlation kernel, one can expect a scaling limit
\begin{equation}
\lim_{n\to\infty}\frac{1}{cn^{1+1/\theta}}K_n\left(\frac{u}{cn^{1+1/\theta}},\frac{v}{cn^{1+1/\theta}}\right)=\mathbb K_{\alpha,\theta}(u,v)
\end{equation}
whenever $V$ is such that the equilibrium measure is one-cut supported with hard edge, and with $c=c_{V,\theta}$ depending on $V$ and $\theta$, as long as $d_1$ given in (\ref{exponent}) and (\ref{c1}) is different from zero.
In the critical regime where the soft edge hits the hard edge, see Section \ref{section: hardsoft}, a new family of kernels can be expected.

\section*{Acknowledgements}
The authors are grateful to G. Akemann for drawing their attention to \cite{LSZ}, to M. Atkin and A. Kuijlaars for useful discussions, and to the anonymous reviewers for their useful suggestions.
They were supported by the European Research Council under the European Union's Seventh Framework Programme (FP/2007/2013)/ ERC Grant Agreement n.\, 307074 and by the Belgian Interuniversity Attraction Pole P07/18. TC was also supported by FNRS.

\end{document}